\documentclass{article}
\usepackage{kr}

\usepackage{times}
\usepackage{soul}
\usepackage{url}
\usepackage[hidelinks]{hyperref}
\usepackage[utf8]{inputenc}
\usepackage[small]{caption}
\usepackage{graphicx}
\usepackage{booktabs}
\usepackage{algorithm}
\usepackage{algorithmic}
\usepackage{amsthm,bm,bbm} 
\usepackage{amsmath} 
\usepackage{amssymb}
\usepackage{mathtools} 
\usepackage{stmaryrd}
\usepackage[dvipsnames]{xcolor} 
 
\usepackage{todonotes} 
\usepackage{xspace}
\usepackage{paralist}
\usepackage[mathscr]{euscript}
\usepackage{enumitem}
\usepackage{subcaption}
\usepackage{tikz} 
\usetikzlibrary{automata, arrows.meta, positioning, shapes}
\usepackage{textcomp}

\usepackage{babel}[..]

\usepackage[appendix=append]{apxproof}

\newtheorem{vb}{Example}
\newtheorem{theorem}{Theorem}
\newtheorem{lemma}{Lemma}
\newtheorem{pr}{Proposition}
\newtheorem{cl}{Claim}

\newtheorem{definition}{Definition}

\newtheorem{corollary}{Corollary}

\newtheoremrep{theorem}{Theorem}
\newtheoremrep{pr}{Proposition}

\newcommand{\A}{\mathcal{G}}
 
\newcommand{\T}{\mathcal{T}}
 
\newcommand{\G}{\mathcal{G}}
 
\newcommand{\C}{\mathcal{C}}

\newcommand{\I}{\mathcal{G}}

\renewcommand{\S}{\mathcal{S}}

\newcommand{\alcio}{$\mathcal{ALCIO}$\xspace}

\newcommand{\alci}{$\mathcal{ALCI}$\xspace}
\newcommand{\alco}{$\mathcal{ALCO}$\xspace}
\newcommand{\alc}{$\mathcal{ALC}$\xspace}

\newcommand{\NC}{N_{C}} 
\newcommand{\NR}{N_{R}} \newcommand{\roles}{\overline{N}_{R}}
\newcommand{\NI}{N_{I}} \newcommand{\NS}{N_{S}}
 
\newcommand{\var}{N_{V}}

\newcommand{\exptime}{\textsc{ExpTime}\xspace}

\newcommand{\trp}{\mathit{tr}^+_{S,\C}}
\newcommand{\trm}{\mathit{tr}^-_{S,\C}}
\newcommand{\trpe}{\mathit{tr}^+_{\emptyset,\C}}
\newcommand{\trpt}{\mathit{tr}^+}
\newcommand{\trmt}{\mathit{tr}^-}
\newcommand{\tr}{\mathit{tr}}
\newcommand{\final}{\mathit{cln}}
\newcommand{\sub}{\mathit{sub}}
\newcommand{\pos}{\mathit{pos}}
\newcommand{\atoms}{\mathit{atoms}}
\newcommand{\wfgc}{\mathit{WF}_{\G,\C}}

\newcommand{\wicgeen}{\mathit{W}_{\G_1,\C}}
\newcommand{\rot}{\mathit{root}}
\newcommand{\semantics}{\ensuremath{\mathbb{S}}\xspace}
\newcommand{\automaton}{\mathbf{A}}

\newcommand{\ptrue}[3]{\lfloor #1 \rfloor^{#2}_{#3}}
\newcommand{\mtrue}[3]{\lceil #1 \rceil^{#2}_{#3}}
\newcommand{\pmtrue}[3]{[ #1 ]^{#2}_{#3}}

\newcommand{\guess}{\mathbf{G}}

\newcommand{\citet}[1]{\citeauthor{#1} (\citeyear{#1})}

\newcommand{\wrt}{w.r.t.\,}

\newcommand*\imge{\raisebox{-0.17\baselineskip}{\includegraphics[height=0.81\baselineskip]{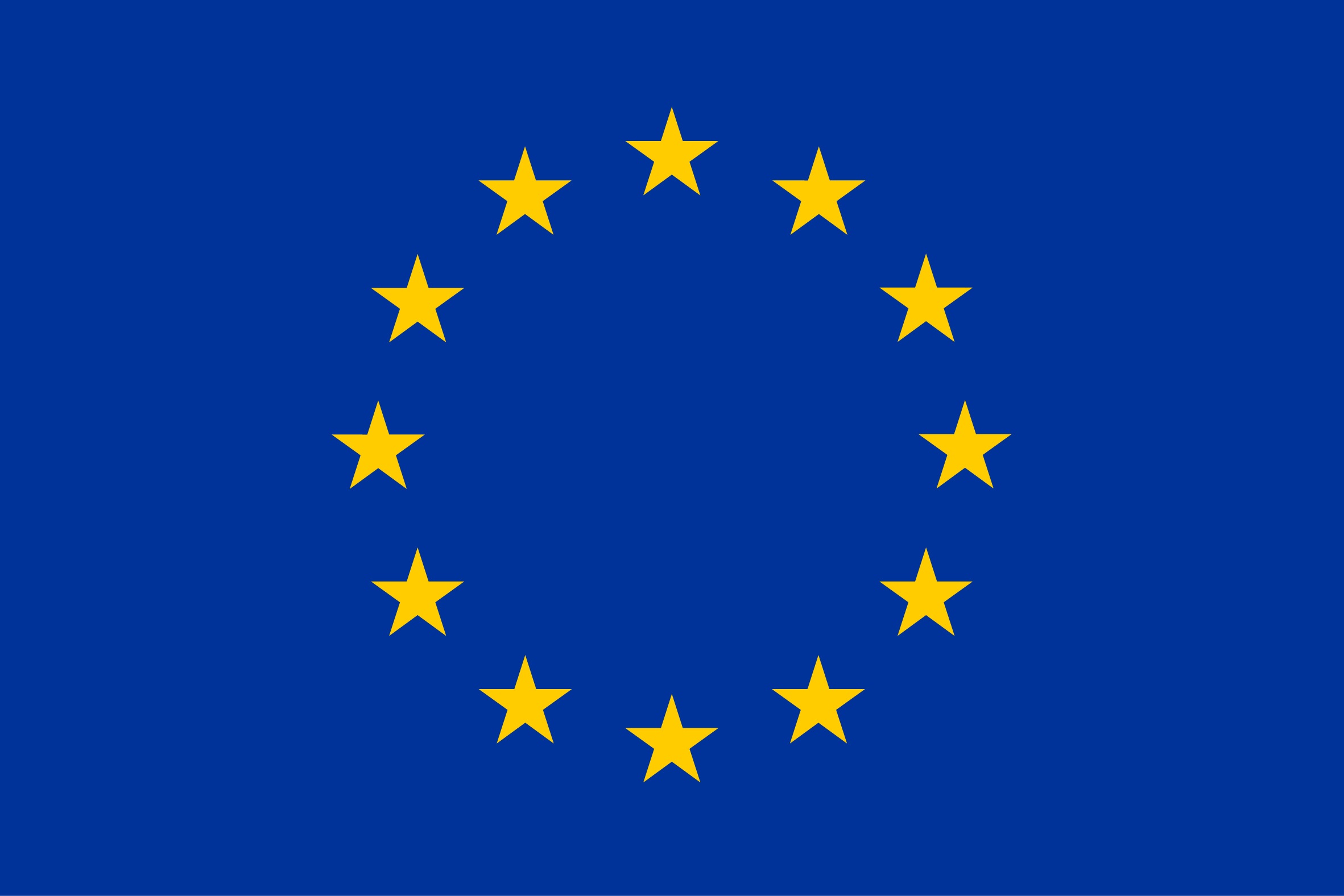}}}

\title{Static Analysis of Recursive SHACL}

\author{%
Anouk Oudshoorn\and
Magdalena Ortiz\and
Mantas \v{S}imkus \\
\affiliations
Institute of Logic and Computation, TU Wien\\
\emails
firstname.lastname@tuwien.ac.at
}

\begin{document}

\maketitle

\begin{abstract}
SHACL (Shapes Constraint Language)
expresses  constraints on RDF data by means of so-called \emph{shapes}. 
Its central service is \emph{validation}: verifying whether a data graph complies with a SHACL document. 
But so far, there are no \emph{static analysis services} to compare documents. 
In this paper, we study the following problem: 
decide whether all graphs that validate one SHACL document also validate another.  Unlike previous works that have considered the implication of shape expressions only, we consider documents comprising (recursive) shape definitions and targets.
We show that implication (a.k.a.\ containment) is undecidable under the supported and the stable model semantics, even for the fragment that uses the description logic $\mathcal{ALCIO}$ for shape expressions. Under the well-founded semantics, in surprising contrast, it is decidable in single exponential time. Our key technical contribution is a translation of SHACL under the well-founded semantics into the full hybrid $\mu$-calculus, revealing a novel link between well-founded models and a fixed point modal logic, and a worst-case optimal automata-based decision procedure.\footnote{This is the long version of this work accepted at KR2026.}
\end{abstract}

\section{Introduction}

\emph{Graph-structured data} is gaining widespread adoption because it
avoids the need for rigid, predefined database schemas, which are
difficult to design in settings where data is highly complex,
incomplete, or subject to frequent structural change. In such
environments, flexible graph data models like Knowledge Graphs (KGs) and
Property Graphs are being increasingly adopted.  To
address the challenges arising in this graph-centric landscape, the
W3C introduced SHACL (Shape Constraint Language) as a standard for
expressing structural and semantic constraints over RDF
graphs~\cite{KK17}. The current SHACL recommendation defines semantics
only for \emph{non-recursive} SHACL constraints. Cyclic 
definitions (i.e., \emph{recursion}) are envisioned in the W3C
recommendation, but no semantics is defined, and validator
implementations are free to handle recursion in an \emph{ad hoc}
manner.  This problem has spurred research into suitable
semantics for SHACL documents with recursive constraints, including the
\emph{supported model semantics}~\cite{DBLP:conf/semweb/CormanRS18}, the
\emph{stable model semantics}~\cite{andresel2020stable}, and the
\emph{well-founded semantics}~\cite{DBLP:conf/kr/OkulmusS24}. As their
names suggest, these SHACL semantics are related to prominent
semantics in Logic Programming. See \citet{ahmetaj2026commonfoundationsrecursiveshape} for a more extensive discussion.

SHACL provides the means to define \emph{shape expressions} that
specify conditions that nodes in a graph are expected to
satisfy. These definitions are paired with a set of \emph{targets} in
a SHACL document. The central service is \emph{validation}, which
consists of verifying whether a graph complies with a SHACL
document. The computational complexity of this task has been studied
for recursive SHACL under various semantics, discovering that in many
settings it is tractable, yet intractability does arise in some
surprising settings, e.g.,\,for \emph{stratified} constraints under the
supported-model semantics~\cite{DBLP:conf/semweb/CormanRS18}. A
natural next step in the study of SHACL is understanding the
computational properties of various \emph{static analysis services},
such as \emph{satisfiability} and \emph{implication
  (a.k.a.\,containment)}. The difference between validation and the main challenge here is that these tasks are \emph{data-independent}.
The first task is to decide whether a given
SHACL document is consistent, i.e., whether there exists at least one data graph that validates it. The second task is to decide whether all graphs that
validate one document must necessarily validate another
document. Checking implication in both directions allows to
decide if two SHACL documents are equivalent. Such services are a
stepping stone towards automated transformations and  optimisation
of SHACL~documents.

Several insights into the
working of SHACL have been achieved by examining its tight connection to
Description Logics
(DLs)~\cite{PARETI2022100721,DBLP:conf/lpnmr/BogaertsJB22,DBLP:conf/wollic/Ortiz23,DBLP:conf/aaai/0001S24}. Specifically,
since unrestricted SHACL shape expressions can be seen as concepts in
very expressive DLs, like ones that support \emph{role-value maps},
undecidability of satisfiability and implication testing in full expressive SHACL can be
inferred~\cite{PARETI2022100721}. Identifying restricted settings in which some static
analysis problems become decidable and to infer upper bounds on
their complexity can be done by drawing from the DL literature. For instance, \citet{DBLP:conf/semweb/LeinbergerSRLS20} provide some positive
results on reasoning about containment of shape
expressions instead of fully-fledged SHACL documents. Furthermore,
\citet{PARETI2022100721} provide some
positive results for satisfiability of recursive SHACL constraints
under the supported-model semantics, while  \citet{dl2025} extends these results further and shows that, for the supported model semantics, full document satisfiability is reducible to satisfiability of a single constraint.
However, our understanding of satisfiability and implication
in recursive SHACL is still very limited. In fact, to the best of our
knowledge, no attempts have been made to encode this problem into decidable logic formalisms, or solve it by other means.
 
Thus, this paper provides the first dedicated study of the computational
complexity of implication of SHACL documents in the presence of recursive
constraints. 
That is, we study the problem of deciding, given two
SHACL documents $\mathcal{S}_1$ and $\mathcal{S}_2$ containing shape
definitions and targets for validation, whether every graph that
validates $\mathcal{S}_1$ also validates $\mathcal{S}_2$. Our first
major result is that, if one considers the supported or the stable
model semantics, this problem is undecidable even when the language
for shape expressions is restricted to the SHACL fragment that
corresponds to the DL $\mathcal{ALCIO}$.  This is somewhat surprising,
but becomes even more intriguing in the light of our second main
result: under the well-founded semantics, this problem is not only
decidable, but is feasible in single exponential time and thus not
harder than deciding satisfiability of a single shape expression.

The main stepping stone to obtaining our positive results is a
translation of SHACL documents under the well-founded semantics into
formulas of the full hybrid $\mu$-calculus.  This translation
establishes a novel connection between logic with least and greatest
fixed points on the one hand, and the well-founded semantics on the
other.  Interestingly, capturing the well-founded semantics requires
limited alternation of fixed point operators.  This connection is, to
our knowledge, novel and of independent interest. We note that some initial works on the well-founded semantics for logic programs have already established a connection for logics with fixed points~\cite{DBLP:journals/jacm/GelderRS91,DBLP:journals/jcss/Gelder93}. Our setting here is  different: their monotonic operator is composed  of two anti-monotonic fixed point operators. That is, the first operator produces an overestimate of the set of negative conclusions, whereas the second is an underestimate. This alternating, estimating behaviour is absent in our work, in the sense that we translate (directly) into (fragments of) the full hybrid $\mu$-calculus. We rely on the insights of our more constructive translation to derive the tight complexity bounds.

That is, our translation of a SHACL document into a $\mu$-calculus formula
unfolds the rules defining shapes into a single formula, potentially
leading to an exponential blow-up. We therefore also provide an
automata-based algorithm for evaluating (the $\mu$-calculus
translation of) a SHACL document that runs in time bounded by a single
exponential function of the size of the input document.  This allows us to
establish optimal complexity bounds for SHACL static validation
services in a range of fragments.

\section{Preliminaries}

\paragraph{Data Graphs.}
Let $\NC,\NR$ and $\NI$ denote countably infinite, mutually disjoint sets of \emph{concept names}, \emph{role names}, and \emph{individuals} respectively, such that $\top \in \NC$. Let $\roles := \{p,p^- \mid p \in \NR\}$ denote \emph{roles}. 
For every $p\in \NR$, let $(p^-)^- = p$.
An \textit{atom} (or, \emph{assertion}) is an expression of the form $A(a)$ or $p(a,a')$, for $A\in \NC$, $p\in \NR$ and $\{a,a'\} \subseteq \NI$. A \emph{data graph} $\A$ is a finite set of atoms. Let $r^\A := \{(a,a') \in \NI \mid r(a,a') \in \A\}$ and $A^\A := \{a \in \NI \mid A(a) \in \A\}$. Let $\Delta^{\G}$ denote the set of all individuals that appear in a data graph $\A$.

\paragraph{Shape Constraint Language (SHACL).}
We introduce SHACL following the line first set out by   \citeauthor{DBLP:conf/semweb/CormanRS18} (\citeyear{DBLP:conf/semweb/CormanRS18}).
We start by defining \textit{shape expressions}. 
We do not define them for full SHACL, but only the fragments considered in this paper. To conveniently talk about these fragments, we borrow notation from Description Logics.
Let $\NS$ denote a countably infinite set of \textit{shape names} that is disjoint from $N_C \cup N_R \cup N_I$.
  Specifically, shape expressions $\varphi$ are build using the following grammar:
\begin{align*}
    \varphi::=  s \mid a\mid A \mid &\;\lnot s \mid
 \varphi \land \varphi \mid \varphi \lor \varphi \mid \forall r.\varphi \mid \exists r.\varphi, 
  \end{align*}
  where $s\in \NS$, $a \in N_I$, $A \in \NC$, 
  and $r \in \roles$. We refer to this fragment as \alcio SHACL, using the DL naming convention. If in the above fragment, $r$ is restricted to $\NR$, that is, we can only access roles in one direction, we call this \alco SHACL. If, on the other hand, the fragment does not use any $a \in \NI$, we are in \alci SHACL. If both restrictions apply, this is referred to as \alc SHACL.
  
  A \textit{shape constraint} is an expression of the form $s \gets \varphi$, for $s \in \NS$ and $\varphi$ a shape expression. We call  $s$ the \textit{head} of the constraint. Given a set of shape constraints $\C$, we assume each shape name $s$ only appears as the head of one constraint; this does not influence expressibility as `$\lor$' may be used in $\varphi$.
  A \textit{shape atom} has the form $s(a)$ with $s \in \NS$ and $a \in \NI$, and we call $\lnot s(a)$
   \textit{negated shape atom}; a \textit{shape literal} is a possibly negated shape atom.

  Let
  \begin{align*}
      \sub(\psi) &:=\{ \psi\} \\
      \sub(\lnot s) &:= \{\lnot s,s\}\\
      \sub(\varphi * \varphi') &:= \{\varphi * \varphi'\} \cup \sub(\varphi) \cup \sub(\varphi') \\
      \sub(\circ r.\varphi) &:= \{\circ r.\varphi\} \cup \sub(\varphi)
  \end{align*}
  for $\psi \in \NC \cup \NI \cup \NS$, $* \in \{\land,\lor\}$ and $\circ \in \{\exists, \forall\}$.
  For any constraint set $\C$, let $\sub(\C) := \{\sub(\varphi) \mid s \gets \varphi \in \C\}$.

A \emph{target} is a pair $(\ell,s)$, where $\ell\in \NI\cup \NC\cup\roles $. A SHACL document is a pair $(\C,\T)$, where $\C$ is a set of shape constraints and $\T$ is a set of targets. With a slight abuse of notation, we sometimes use a shape atom $s(\ell)$ to specify a target  $(\ell,s)$, and write $(\C,s(\ell))$ instead of $(\C,\{(\ell,s)\})$. 

The semantics of SHACL is based on the notion of a \emph{shape assignment}, defined as a set $S$ of shape shape literals that are not contradictory, that is, there are no $\{s(a),\lnot s(a)\} \subseteq S$. 
We view such shape assignments as three-valued: 
intuitively, $s(a) \in S$ means that $s$ is validated at $a$, $\lnot s(a) \in S$ means that $s$ is not validated, and $\{s(a), \lnot s(a)\} \cap S = \emptyset$ means that the validation of $s$ at $a$ is \textit{undefined} in $S$. We say $S$ is \emph{total}, if for every shape name $s$ and individual $a$ that appear in $S$, we have $s(a)\in S$ or $\lnot s(a) \in S$.

We say $\G$  validates a set of targets $\T$ under a shape assignment $S$ if the following are satisfied:
\begin{enumerate}
\item  $s(a) \in S$, for every $(a,s)\in \T$,
\item  $s(a) \in S$, for every $(A,s)\in \T$ and $a \in A^{\A}
  $, and 
\item  $s(a) \in S$, for every $(r,s)\in \T$ and $(a,a') \in r^{\A}
$.  
\end{enumerate}
To define validation for a SHACL document $(\C,\T)$, we still need 
to give semantics  to shape expressions under a given shape assignment. For recursive SHACL, three ways to do this have been advocated in the literature: the well-founded \cite{DBLP:conf/kr/OkulmusS24}, supported \cite{DBLP:conf/semweb/CormanRS18}, and stable semantics \cite{andresel2020stable}. 
The well-founded semantics yields a unique assignment for each $\G$ and $\C$. In contrast, the supported and the stable semantics both yield a set that may contain several \emph{total} assignments (or none at all); in such cases, one adopts \emph{brave} validation. 

\begin{definition}
A \emph{semantics} $\semantics$ for SHACL is a function that maps each pair of a graph and a set of constraints to a set of assigments. We say that $\G$ satisfies $(\C,\T)$ under \semantics if $\G$ validates $\T$ under some assignment in $\semantics(\G,\C)$. We say that a node $a$ in $\G$ satisfies a shape $s$ w.r.t.\,$\C$ under \semantics if  $\G$ satisfies $(\C,s(a))$.
If some node of $\G$ satisfies $s$ w.r.t.\,$\C$, we say that \emph{$\G$ 
satisfies $s$ w.r.t.\,$\C$}, or for simplicity, just \emph{$\G$ 
satisfies $(\C,s)$}. 
\end{definition}

We define next the well-founded semantics, which is the focus of this paper, and briefly recall the supported semantics. For the full definition of the stable semantics, we refer the reader to \cite{andresel2020stable}. 

Recall that, based on the allowed syntax for SHACL documents, concrete individuals may explicitly appear inside shape expression or targets. To make presentation simpler, when we talk about checking if a graph $\G$ satisfies a document $(\C,\T)$, we assume that $\G$ is \emph{compatible} with  $(\C,\T)$. By this we mean that all individuals that appear in $(\C,\T)$ also appear in $\G$.  Validation for incompatible  $\G$ and $(\C,\T)$ is considered undefined. When we quantify over possible graphs $\G$ for validation or non-validation w.r.t.\,a document $(\C,\T)$, the quantification is always limited to graphs that are compatible with  $(\C,\T)$.

\paragraph{Well-founded Semantics.}

\begin{figure*}[t]
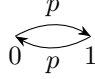
\centering
  \begin{align*}
    \ptrue{s}{\G}{S} &:= \{a\mid s(a)\in S\}  &    \mtrue{s}{\G}{S} &:= \{a\in \Delta^{\G}\mid \neg s(a)\not\in S\} \\[.5ex]    \ptrue{\neg s}{\G}{S} &:= \{a\mid \neg s(a)\in S\} &
    \mtrue{\neg s}{\G}{S} &:= \{a\in \Delta^{\G}\mid s(a)\not \in S\} \\[.5ex]
           \pmtrue{a}{\G}{S} &:= \{ a^{\G} \} &
    \pmtrue{A}{\G}{S} &:= A^{\G}  & ~~~~~ [\cdot]\in \{\ptrue{\cdot}{}{},\mtrue{\cdot}{}{}\big\}   \\
    \pmtrue{\varphi \lor \varphi'}{\G}{S} &:= \pmtrue{\varphi}{\G}{S}\cup \pmtrue{\varphi'}{\G}{S} & \pmtrue{\varphi \land \varphi'}{\G}{S} &:=\pmtrue{\varphi}{\G}{S}\cap \pmtrue{\varphi'}{\G}{S}   &  [\cdot]\in \{\ptrue{\cdot}{}{},\mtrue{\cdot}{}{}\big\}  
  \end{align*}
  \vspace*{-16pt}
  \begin{align*}
     \;\;\;\;\;\:\:\:\pmtrue{\forall r.s}{\G}{S} &:=\big\{ a\in \Delta^{\G}\mid \forall a'\in \Delta^{\G}: (a,a')\in r^{\G} \mbox{ implies } a'\in \pmtrue{s}{\G}{S}   \}\qquad \qquad    &  [\cdot]\in \{\ptrue{\cdot}{}{},\mtrue{\cdot}{}{}\big\}  \\[.5ex]
    \pmtrue{\exists r.s}{\G}{S} &:=\big\{ a\in \Delta^{\G}\mid \text{ exists } a' \text{ s.t. } (a,a')\in r^{\G} \mbox{ and }a'\in \pmtrue{s}{\G}{S}   \}   &  [\cdot]\in \{\ptrue{\cdot}{}{},\mtrue{\cdot}{}{}\big\} 
  \end{align*}
  
  \caption{Evaluating shape expressions: upper and lower bounds.}
  \label{fig:expr-eval}
\end{figure*}

For a given data graph $\A$ and shape assignment $S$,  we define two functions 
$\ptrue{\cdot}{\A}{S}$ and $\mtrue{\cdot}{\A}{S}$  that map shape expressions to sets of nodes as described in Figure \ref{fig:expr-eval}.  
Let $\C(s) := \varphi$, where $s \gets \varphi \in \C$ is the unique constraint in $\C$ with head $s$.  
Intuitively, assuming that the literals in $S$ are true, $\ptrue{\varphi}{\A}{S}$ is the set of nodes where $\varphi$ is certainly validated, 
whereas $\mtrue{\varphi}{\A}{S}$ is the set of nodes for which validation is possible, as it is not prevented by $S$. 
We can use $\ptrue{\cdot}{\A}{S}$  to infer positive shape literals in the well-founded model: if $a \in \ptrue{\varphi}{\A}{S}$, we can infer $s(a)$. In contrast, from $\mtrue{\cdot}{\A}{S}$ we can infer negated shape literals: if $a \not\in \mtrue{\varphi}{\A}{S}$, then we can infer $\lnot s(a)$. 
 These inferred literals are added to the well-founded model incrementally, using the following one-step operators.     
For the positive inferences, the operator is
$$
T_{\A,\C}(S) := \{s(a) \mid s \gets \varphi \in \C, a \in \ptrue{\varphi}{\A}{S}\}.
$$

For inferring negative atoms, we rely on 
\emph{unfounded sets}.
For a set $U$ of shape atoms, we let $\lnot.U := \{\lnot s(a) \mid s(a) \in U\}$.
Given a shape assignment $S$, a data graph $\A$ and a set of constraints $\C$, we call a set of shape atoms $U$ \textit{unfounded} w.r.t. $S$, $\A$ and $\C$ if $a \not \in \mtrue{\C(s)}{\A}{S \cup \lnot.U}$ for all $s(a) \in U$. As being unfounded is preserved under union, there exists a unique \textit{greatest unfounded set} w.r.t.\ $S$, $\A$ and $\C$ that is not subsumed by a strictly larger unfounded set w.r.t.\ $S$, $\A$ and $\C$. Let $U_{\A,\C}$ be a function that maps each shape assignment $S$ to the greatest unfounded set w.r.t.\ $S$, $\A$ and $\C$.

Both positive and negative consequences are combined in the operator $W_{\A,\C}$. 
Like $T_{\A,\C}$ and $U_{\A,\C}$, it is a function that maps shape assignments to shape assignments.
\begin{equation}\label{eq:wfs-def}
W_{\A,\C}(S) := T_{\A,\C}(S) \cup \lnot . U_{\A,\C}(S).
\end{equation}
We say a function $f$ is \textit{monotone} if $x \subseteq y$ implies $f(x) \subseteq f(y)$. The above \textit{immediate consequence operator} $W_{\A,\C}$ is a monotone function, hence it has a unique least fixed point. We denote this fixed point $\wfgc$ and call it the \emph{well-founded model} of $\G$ and $\C$.
 
\begin{vb}\label{vb:mantasexample1}
    Consider the constraint set
    $\C = \{s \gets A \lor \exists p.\lnot t, t \gets \lnot s \lor \exists r. t\}$. The SHACL document $(\C,\{s(0)\})$ is satisfiable in, for instance, the following structure $\G_1$.
\newcommand*{\nodes}{5}

    \begin{center}
	\begin{tikzpicture}
        		\foreach \i in {0,...,\nodes} {
			\draw[-Stealth] (\i,0) to node[above] {$p$} (\i+1,0);
            \draw[-Stealth] (\i+1,0) to [loop above, edge node={node [above] {$r$}}] (\i+1,0);
            \draw (0,0) node at (\i,-0.25) {$\i$};
		}
		\draw (0,0) node at (\nodes+1.2,0.1) {$A$};
        \draw (0,0) node at (\nodes+1,-0.25) {$6$};
    	\end{tikzpicture}
        \end{center}
    To see this, observe the following step-by-step computation of the least fixed point $\mathit{WF}_{\G_1,\C}$.
    \begin{align*}
        \wicgeen(\emptyset) &= \{s(6)\} \cup \emptyset\\
        \wicgeen^2(\emptyset) &= \{s(6)\} \cup \{\lnot t(6)\}\\
        \wicgeen^3(\emptyset) &= \{s(6),s(5)\} \cup \{\lnot t(6)\}\\
        \wicgeen^4(\emptyset) &= \{s(6),s(5)\} \cup \{\lnot t(6),\lnot t(5)\}\\
        &\qquad\qquad\cdots
    \end{align*}
    Note that if all $p$-paths from 0 to an $A$ are infinitely long (and there exists one), this least fixed point will not terminate in a finite number of steps. 
\end{vb}

\paragraph{Supported Model Semantics.}
A supported model for 
constraint set $\C$ and a graph $\I$ is a total shape assignments $S$ that
 satisfies 
$s^{\I,S} = \ptrue{\varphi}{\I}{S}$ for every
$s \gets \varphi \in \C$.  We say $\I$ validates a SHACL document $(\C,\T)$ under the \emph{Supported Model Semantics}, if there is some supported model $S$ of $\C$ and a graph $\I$ that
satisfies $\T$. Note that above we used $\ptrue{\cdot}{\G}{S}$. However, both $\ptrue{\cdot}{\G}{S}$ and $\mtrue{\cdot}{\G}{S}$ can be used interchangeably here, since $\ptrue{\varphi}{\I}{S}=\mtrue{\varphi}{\I}{S}$, for any
$\varphi$, $\I$, and total $S$.

\begin{vb} 
Consider the same constraint set as in Example \ref{vb:mantasexample1}, but now for the following structure $\G_2$.
\begin{center}
\begin{tikzpicture}
    \draw[-Stealth] (0,0) to [out=30,in=150] node[above] {$p$} (1,0);
    \draw[-Stealth] (1,0) to [out=210,in=330] node[below] {$p$} (0,0);
    \draw (0,0) node at (0,-0.25) {$0$};
    \draw (0,0) node at (1,-0.25) {$1$};
\end{tikzpicture}
\end{center}
As $\mathit{WF}_{\G_2,\C}(\emptyset) = \emptyset$, we find that $\G_2$ and $\C$ only have empty well-founded models, whereas there are two supported models $S_1$ and $S_2$: 
\begin{align*}
    S_1 &= \{s(0),s(1),\lnot t(0),\lnot t(1)\}\\
    S_2 &= \{\lnot s(0), \lnot s(1), t(0), t(1)\}.
\end{align*}
\end{vb} 

\paragraph{Stratified Constraints.}

Derived from the well-known class of stratified programs \cite{10.5555/61352.61354}, stratified SHACL can be defined in the following way.

\begin{definition}
  We say a shape name $s$ is \emph{defined} in a set $\C$ of constraints if $s\gets \varphi\in \C$ for some $\varphi$.
  A set $\C$ of
  constraints is \emph{stratified} if it can be partitioned into sets
  $\C_0,\ldots,\C_k$ such that, for all $0 \leq i \leq k$, the
  following hold:

  \begin{enumerate}
  \item If $i< k$ and $s' \in \sub(\varphi)$ for some
    $s\gets \varphi \in \C_i$, then $s'$ is defined in
    $\C_{0}\cup\ldots \cup \C_i$.
  \item If $\lnot s' \in \sub(\varphi)$ for some
    $s\gets \varphi \in \C_i$, then $s'$ is defined in
    $\C_0\cup\ldots \cup \C_{i-1}$.
  \end{enumerate}
  A set of constraints is \emph{stratified} if it admits a
  stratification.
\end{definition}

For stratified constraints, $\wfgc$  is total, it is a supported model, and it is the unique stable model of $\G$ and $\C$; it is sometimes called the \emph{least} or \emph{perfect} model of $\G$ and $\C$. 

\paragraph{Normal Form.}
To simplify presentation, without loss of generality, we will mostly focus on SHACL constraints in which there is no nesting of constructors. Specifically, we focus on constraints that have one of the following forms:
\begin{center}
\begin{tabular}{l l l l}
    $s \gets A$ & $s \gets a$ & $s \gets \lnot s'$   & $s \gets \exists r.s'$ \\[0.5ex]
    $s \gets \forall r.s'$ & $s \gets s' \land s''$ & $s \gets s' \lor s''$
\end{tabular}
\end{center}
where $A\in \NC$, $a\in \NI$, $r\in \roles $, and $\{s,s',s''\} \subseteq \NS$. For all semantics we consider, transforming an arbitrary SHACL document into one where all constraints are in normal form while preserving validation is straightforward. Intuitively, this is because whenever a complex subformula $\psi$ occurs in a shape constraint $s\gets \varphi$, we can simply replace $\psi$ in $\varphi$ by a fresh shape name $s'$ and add  $s'\gets \psi$ to the constraint collection. See, e.g.\ the constructions by \citet{andresel2020stable} and \citet{DBLP:conf/dlog/Oudshoorn0S24}. 

\section{Static Analysis of SHACL Specifications}

In this paper we study the following decision problems  for the different semantics $\mathbb{S}$ described above.

\begin{description}
    \item[(Finite) document satisfiability:] Given
    a SHACL document $(\C,\T)$, decide whether there exists a (finite) graph  that satisfies $(\C,\T)$ under \semantics. 
       \item[Finite model property:] Is it the case that, under \semantics, every satisfiable SHACL document is finitely satisfiable?  
        \item[(Finite) document implication:] Given two SHACL documents $(\C_1,\T_1)$ and $(\C_2,\T_2)$, decide whether every (finite) graph that satisfies $(\C_1,\T_1)$ under \semantics also satisfies $(\C_2,\T_2)$.
    \end{description}
    We also consider the special cases where we target validation of one specific shape name:  
    \begin{description}
    \item[(Finite) shape satisfiability \wrt constraints:]  
    Given a set of  constraints $\C$ and a shape name $s$, decide whether there exists a (finite) graph that 
    satisfies $(\C,s)$ under \semantics.
        \item[(Finite) shape implication \wrt constraints:] Given shape constraints $\C_1$ and $\C_2$ and
     shape names $s_1$ and $s_2$, decide whether every node that  satisfies $(\C_1,s_1)$ under \semantics in a graph $\G$ also satisfies $(\C_2,s_2)$.
\end{description}

To our knowledge, document and shape satisfiability are the only problems that have been studied before, except for one work that studied implication (sometimes called containment) for SHACL, but in a more restricted setting; 
\citeauthor{DBLP:conf/semweb/LeinbergerSRLS20}
(\citeyear{DBLP:conf/semweb/LeinbergerSRLS20}) studied satisfiability and implication problem for shape expressions alone, with no constraints. 
Their problem corresponds to deciding whether 
$(\{s \gets \varphi_1\},s)$ is satisfiable, and whether
$(\{s_1 \gets \varphi_1\},s_1)$ implies $(\{s_2 \gets \varphi_1\},s_2)$ for given $\varphi_1,\varphi_2$.

For the supported model semantics, document satisfiability has been considered by \citeauthor{PARETI2022100721} (\citeyear{PARETI2022100721}) and 
\citeauthor{dl2025} (\citeyear{dl2025}).
In this semantics, these problems naturally reduce to satisfiability and logical implication in fragments of classical first-order logic (FO), or equivalently, to concept satisfiability and subsumption w.r.t.\ TBoxes in Description Logics (DLs). \citeauthor{PARETI2022100721} obtained many undecidability and some decidability results by analysing the target FO fragments, which were then significantly expanded and tightened in terms of complexity by \citeauthor{dl2025} building on DLs.  
Some of the results mentioned above consider both finite and unrestricted models; others focus on finite ones. In some cases, the finite model property is inferred by establishing a connection to logics that enjoy it.  We refer to \citeauthor{dl2025} (\citeyear{dl2025}) for a more extensive discussion.

We are not aware of any works considering satisfiability and containment under the stable or well-founded semantics, but an \exptime upper bound for shape satisfiability \wrt constraints under the stable semantics for the $\mathcal{ALCI}$ fragment can be inferred from the work of
\citet{DBLP:conf/aaai/0001S24} on DL terminologies under the stable (a.k.a.\ \emph{equilibrium}) semantics.  

\definecolor{ashgrey}{rgb}{0.82, 0.84, 0.81}
\newcommand*{\xMin}{0}%
\newcommand*{\xMax}{2}
    \begin{figure}
        \centering
        \begin{tikzpicture}
            \foreach \i in {\xMin,...,\xMax} {
                \foreach \j in {\xMin,...,\xMax} {
                    \draw[-Stealth] (\i,\j) -- node[below] {$r$} (\i+1,\j);
                    \draw[-Stealth] (\i,\j) -- node[left] {$u$} (\i,\j+1);
                }
                \draw[-Stealth] (\i,\xMax+1) -- node[below] {$r$} (\i+1,\xMax+1);
                \draw[-] (\i,\xMax+1) -- (\i,\xMax+1.4);
                \draw[-Stealth] (\xMax+1,\i) -- node[left] {$u$} (\xMax+1,\i+1);
                \draw[-] (\xMax+1,\i) -- (\xMax+1.4,\i);
                \draw[-Stealth,densely dashed,ashgrey] (\i+1,0) to [out= 225,in=300] (0,0); 
                \draw[-Stealth,densely dashed,ashgrey] (0,\i+1) to [out= 225,in=150] (0,0);
                \draw[-Stealth,densely dashed,ashgrey] (\i+1,1) to [out= 260,in=300] (0,0); 
            }
            \foreach \i in {1,...,\xMax} {
                \draw[-Stealth,densely dashed,ashgrey] (1,\i+1) to [out=190,in=150] (0,0);
                \draw[-Stealth,densely dashed,ashgrey] (\i+1,2) to [out= 262,in=300] (0,0);
            }
            \foreach \i in {2,...,\xMax} {
                \draw[-Stealth,densely dashed,ashgrey] (2,\i+1) to [out=188,in=150] (0,0);
                \draw[-Stealth,densely dashed,ashgrey] (\i+1,3) to [out= 264,in=300] (0,0);
            }
                \draw[-Stealth,densely dashed,ashgrey] (0,0) to [out=260,in=280,loop] (0,0);
            
            \draw[-] (\xMax+1,\xMax+1) -- (\xMax+1.4,\xMax+1);
            \draw[-] (\xMax+1,\xMax+1) -- (\xMax+1,\xMax+1.4);
            
            \draw (0,0) node at (-0.4,-0.2) {$s_1,a$};
            
        \end{tikzpicture}
        \caption{Infinite grid that, after adding $s$ as label to every node, shows undecidability of shape implication \wrt constraints under the supported model semantics. The grey dashed arrows indicate the extension of the role $l$.}
        \label{fig:gridwithdiagonalE}
    \end{figure}
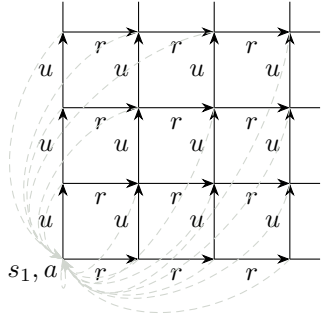

\begin{theorem}\label{thm:undecSupp} 
Shape implication w.r.t.\ constraints for (stratified) \alcio SHACL is undecidable, under the supported model semantics.
\end{theorem}

\begin{proof}
    We prove this by reducing the tiling problem, a well-known undecidable problem \cite{berger1966undecidability}, to our setting. To achieve this reduction we construct two shape pairs such that $(\C,s_1)$ implies $(\C',s')$ iff there exists a grid that can be decorated with a tiling $t: \mathbb{N} \times \mathbb{N} \rightarrow D$. 

    That is, let $\C = \{s_1 \gets a \land s, s \gets \exists l.a \land \exists u.s \land \exists r.s \land \bigvee_{d \in D}s_d \land \lnot s_{\bot}\} \cup \C_{\mathit{tiles}}$, where $\C_{\mathit{tiles}}$ contains the constraints ensuring the correct horizontal and vertical relations between the tiles, and sends the error shape $s_{\bot}$ back in case more than one tile was assigned to one node - as this is straightforward to construct, we will focus on the grid construction in the rest of this proof.
    Furthermore, we let $\C' = \{s_a \gets s_a, s_b \gets s_b, s' \gets \exists l^-.(\exists r\exists u.s_a \land \exists u.\exists r.s_b) \land \forall l^-.((s_a \lor s_b) \land \lnot (s_a \land s_b)\}$.

    Deciding whether $(\C,s_a)$ implies $(\C',s')$ is equivalent to deciding whether there exists a data graph $\G$ such that $\G$ satisfies $(\C,s_a)$, and not $(\C',s')$. Note that to satisfy $(\C,s_a)$ it must be possible to map an infinite binary tree (with $u$ and $r$ edges) into $\G$, with every node linking back via $l$ to the root $a$. For $\G$ to not satisfy $(\C',s')$ means it is not possible to properly decorate $\G$ with shape names. If we take a closer look at $\C'$, it is immediate that $s_a \gets s_a$ and $s_b \gets s_b$ cannot cause any problem. Therefore, the impossibility must lie in $s' \gets \exists l^-.(\exists r\exists u.s_a \land \exists u\exists r.s_b) \land \forall l^-.(s_a \lor s_b \land \lnot (s_a \land s_b)$. For a node to not be an $s'$ means that at least $\exists l^-.(\exists r\exists u.s_a \land \exists u\exists r.s_b)$ or $\forall l^-.((s_a \lor s_b) \land \lnot (s_a \land s_b)$ cannot be realised on the given structure. As the latter, an exclusive-or labelling nodes with either an $s_a$ or an $s_b$, cannot be blocked in any structure, it is $\exists l^-.(\exists r\exists u.s_a \land \exists u.\exists r.s_b)$ that must be blocked. As every node must have a $u$- and $r$-successor, and as $s_a$ and $s_b$ are mutually exclusive, this can only be blocked if for every node reachable by $l^-$, we find there can be only one node reachable by both $ru$ or $ur$. In other words, $(\C,s_a)$ does not imply $(\C',s')$ iff there exists a grid that can be decorated with tiles.
\end{proof}

Note that the problem remains undecidable even if we require the set of constraints to be the same in both documents: in the proof, we can replace both $\C$ and $\C'$ by their union.
Not surprisingly, the undecidability applies also to the stable semantics, but only in the non-stratified setting.

\begin{theorem}
    Shape implication w.r.t.\ constraints for \alcio SHACL is undecidable, under the stable model semantics.
\end{theorem}

To see this, update $\C'$ to $\C' = \{s_a \gets \lnot s_b, s_b \gets \lnot s_a, s' \gets \exists l^-.(\exists r\exists u.s_a \land \exists u\exists r.s_b)\}$. Here, we use the non-stratifiedness of $\C'$ to enforce the choice of an $s_a$ or $s_b$ label. The rest of the proof works in the same way.

Thus, we shift our attention to SHACL containment under the well-founded semantics in the rest of this article, and find we can recover decidability of containment. We note it will sometimes be convenient to treat implication as a special case of containment. 

\begin{pr}\label{pr:docsat-docimp}
A document $(\C,\T)$ is (finitely) satisfiable under \semantics iff $(\C,\T)$ does not (finitely)
  imply $(\emptyset,\emptyset)$ under \semantics.
\end{pr}

\section{Full Hybrid $\mu$-calculus}

We first provide some preliminary definitions for the full hybrid $\mu$-calculus.\footnote{$\mu$-calculus with inverse programs and nominals; we follow the naming convention of \cite{DBLP:journals/lmcs/BonattiLMV08}.}
Let $\var$ be a set of \textit{variables}, disjoint from $\NI, \NC$ and $\NR$.
Then, a full hybrid $\mu$-calculus formula is given by
\begin{align*}
    \Phi ::= A \mid \lnot A \mid a \mid \lnot a& \mid X \mid \Phi \land \Phi \mid \Phi \lor \Phi \mid \\
    &[r] \Phi \mid \langle r \rangle \Phi \mid \mu X.\Phi \mid \nu X.\Phi,
\end{align*}
for $A \in \NC$, $a \in \NI$, $X \in \var$, $r \in \roles$. Let $\top := A \lor \lnot A$, and $\bot := A \land \lnot A$ for a $A \in \NC$.

The set of \textit{subformulas} is defined recursively by setting $\sub(C) := \{C\}$ and $\sub(\lnot C) := \{\lnot C\}$ for each $C \in \NI \cup \NC \cup \var$, $\sub(\Phi * \Phi'):= \{\Phi * \Phi'\} \cup \sub(\Phi) \cup \sub(\Phi')$, for each $* \in \{\land, \lor\}$, $\sub(@\Phi) := \{@\Phi\} \cup \sub(\Phi)$ for each $@ \in \{\langle r\rangle,[ r ] \mid r \in \roles\}$, and $\sub(\sigma X.\Phi) := \{\sigma X.\Phi\} \cup \sub(\Phi)$, for $\sigma \in \{\mu,\nu\}$. We say $\Psi$ appears or is contained in $\Phi$ is $\Psi \in \sub(\Phi)$.

A variable $X \in \var$ is a \textit{free variable} if $X$ appears outside of the scope of a fixed point operator $\sigma X$, for $\sigma \in \{\mu,\nu\}$. A formula is called \textit{closed} if it contains no free variables. 

Given an data graph $\I$, the semantics $||\Phi||^\G_V \subseteq \Delta^\I$ of a closed formula $\Phi$ is defined as in Figure \ref{fig:mucalcsemantics}. 
Here, $V: \var \rightarrow 2^{\Delta^\G}$ is a function that maps each fixed point variable $X$ to a set of nodes. During evaluation, the function $V$ is updated as follows: 
$V[X \rightarrow U](X) := U$, and $V[X \rightarrow U](Y) := V(Y)$
for $Y \not = X$.

\begin{figure*}[t]
    \begin{align*}
    ||\Phi \land \psi||^\I_V &:= ||\Phi||^\I_V \cap ||\psi||^\I_V & ||A||^\I_V &:= A^\I\\
    ||\Phi \lor \psi||^\I_V &:= ||\Phi||^\I_V \cup ||\psi||^\I_V &
    ||\lnot A||^\I_V &:= \Delta^\I \setminus A^\I\\
    ||[ r ] \Phi||^\I_V &:= \{e \in \Delta^\I \mid \text{ if } (e,e') \in r^\I \text{, then }  e' \in ||\Phi||^\I_V\}&
    ||a||^\I_V &:= a^\I\\
    ||\langle r \rangle \Phi||^\I_V &:= \{e \in \Delta^\I \mid \text{ there is } e' \text{ s.t. } (e,e') \in r^\I, e' \in ||\Phi||^\I_V\} &
    ||\lnot a||^\I_V &:= \Delta^\I \setminus a^\I\\
    ||\mu X. \Phi||^\I_V &:= \bigcap \{U \subseteq \Delta^\I : ||\Phi||_{V[X \rightarrow U]} \subseteq U\} &
    ||X||^\I_V &:= V(X)\\
    ||\nu X. \Phi||^\I_V &:= \bigcup \{U \subseteq \Delta^\I : U \subseteq ||\Phi||_{V[X \rightarrow U]}\}
\end{align*}
    \caption{Semantics of $\mu$-calculus formulas.}
    \label{fig:mucalcsemantics}
\end{figure*}

For a closed formula $\Phi$, we may use $\I,c \models \Phi$ as a shorthand for $c \in ||\Phi||^\I_V$. Note that because each variable $X$ always occurs in the scope of a fixed point operator $\sigma X$ in a closed formula, it is unnecessary to provide a valuation function $V$.

Although unusual, we allow fixed point operators $\sigma X.\Psi$ such that $X$ does not occur in $\Psi$. We call a formula \textit{clean} if there do not appear such fixed point operators in the formula.
Note that if $X$ is not contained in $\Psi$, removing or adding $\sigma X$ does not influence the semantics of the formula. 
\begin{toappendix}
Every formula can be \textit{cleaned}, as defined recursively by the function $\final$, given in Figure \ref{tab:cleanise}, with the main advantage of improving the readability of the formula.

\begin{figure}[h]
    \begin{align*}
        \final(\Phi \land \Phi') &:= \final(\Phi) \land \final(\Phi') \\
        \final(\Phi \lor \Phi') &:= \final(\Phi) \lor \final(\Phi') \\
        \final(\langle r\rangle \Phi) &:= \langle r \rangle \final(\Phi) \qquad \final(A) := A \\ 
        \final([ r] \Phi) &:= [r] \final(\Phi) \qquad\;\, \final(a) := a \\
        \final(\lnot \Phi) &:= \lnot \final(\Phi) \qquad\;\, \final(X) := X \\ 
        \final(\sigma X.\Phi) &:= \begin{cases}
            \final(\Phi) & \text{if } X \in \sub(\Phi)\\
            \sigma X.\final(\Phi) & \text{if } X \not \in \sub(\Phi)
        \end{cases}
    \end{align*}
    \caption{Cleaning function.}
    \label{tab:cleanise}
\end{figure}
\end{toappendix}

\begin{pr}
    For each fully hybrid $\mu$-calculus formula, and each data graph $\I$ and node $c \in \NI$, we find
    $$
    \G,c \models \Phi \quad\text{ iff }\quad \G,c \models \final(\Phi).
    $$
\end{pr}

A well-known classical result is that fixed points can be calculated by iteratively updating the valuation for $X$ with the result of the previous approximation (see, e.g.\ \citet{rudimentsofmucalc2001}). That is, the (non-transfinite version of the) approximations of the least and greatest fixed point are given by 
\begin{align*}
    ||\mu^0X.\Phi||^\I_V &:= \emptyset,\\
    ||\mu^{\alpha+1}X.\Phi||^\I_V &:= ||\Phi||^\I_{V[X\mapsto ||\mu^\alpha X.\Phi||^\I_V]}\\
    ||\mu^{\omega}X.\Phi||^\I_V &:= \bigcup_{\alpha<\omega}||\mu^\alpha X.\Phi||^\I_V,
\end{align*}
and
\begin{align*}
    ||\nu^0X.\Phi||^\I_V &:= \Delta^\G,\\
    ||\nu^{\alpha+1}X.\Phi||^\I_V &:= ||\Phi||^\I_{V[X\mapsto ||\nu^\alpha X.\Phi||^\I_V]},\\
    ||\nu^{\omega}X.\Phi||^\I_V &:= \bigcap_{\alpha<\omega}||\nu^\alpha X.\Phi||^\I_V.
\end{align*}

These approximations lead to the following infinite increasing resp. decreasing chains
\begin{align*}
    ||\mu^0X.\Phi||^\I_V &\subseteq ||\mu^1X.\Phi||^\I_V \subseteq \ldots \subseteq ||\mu^\alpha X.\Phi||^\I_V \subseteq \ldots\\
    ||\nu^0X.\Phi||^\I_V &\supseteq ||\nu^1X.\Phi||^\I_V \supseteq \ldots \supseteq ||\nu^\alpha X.\Phi||^\I_V \supseteq \ldots
\end{align*}
Moreover, as we assume $|\Delta^\G| \leq \omega$, we can conclude the following.

\begin{pr}
    For all full hybrid $\mu$-calculus formulas $\sigma X.\Phi$, we find
    \begin{align*}
        ||\sigma^{\omega}X.\Phi||^\I_V = ||\sigma X. \Phi||^\I_V.
    \end{align*}
\end{pr}

\section{Translating SHACL into $\mu$-calculus}

The goal of this section is to provide a translation that takes as input a SHACL specification 
$\C$ and target shape name $s$ and produces a $\mu$-calculus formula $\Phi_{\C,s}$ such that, for every graph $\G$, the nodes that validate $s$ with respect to $\C$ are exactly the nodes that satisfy $\Phi_{\C,s}$. 

The translation is presented in Figure~\ref{fig:translationfunctions}  
and uses two mutually recursive functions, $\trp$ and $\trm$. The desired $\Phi_{\C,s}$ will be the (cleaned version of) $\trpe(s)$. 
Recall that the well-founded model is defined as the least fixed point of the operator $W_{\A,\C}$ in (\ref{eq:wfs-def}); it is thus not surprising that $\trpe(s)$ takes the form of a least fixed point $\mu X \varphi$. 
The positive part $\trp$ will mimic the derivation of the positive atoms that are introduced by $T_{\A,\C}(S)$ during the well-founded model computation, while the 
$\trm$ part will be the counterpart of the negative atoms added by $\lnot . U_{\A,\C}(S)$.
It is this second part that is more interesting. 
Recall that $U_{\A,\C}(S)$ is the \textit{greatest} unfounded set, that is, the largest set of atoms that do not break a certain condition ($a \not \in \mtrue{\C(s)}{\A}{S \cup \lnot.U}$ for all $s(a) \in U$).
This seems to naturally correspond to a greatest fixed point, but we need a bit more than just a  $\nu$: note that when deciding whether $b \in \mtrue{s}{\G}{S \cup \lnot.U}$, we do not just check whether 
$\lnot s(b) \not \in \lnot. U$, but also whether 
$\lnot s(b) \not \in S$, that is, the extension of the set is affected not only by the set $U$ that is being computed as a greatest fixed point, but also by the set $S$ that are computed by the least fixed point corresponding to $W_{\A,\C}(S)$. Therefore, our translation produces a formula with limited alternation: it takes the form $\mu X.\varphi(X)$, where $\varphi(X)$ may contain subformulas of the form $\nu Y.\psi(X,Y)$. There are, however, no fixed points nested within greatest fixed points: observe that the subscript $S$ is changed to $\pos(S)$ when we move to the negative part $\trp$ of the translation.

\begin{figure*}[t]
    \begin{align*}
        \trp(s \land s') &:= \trp(s) \land \trp(s') &
        \trm(s \land s') &:= \trm(s) \lor \trm(s') \\
        \trp(s \lor s') &:= \trp(s) \lor \trp(s') &
        \trm(s \lor s') &:= \trm(s) \land \trm(s') \\
        \trp(\exists r.s) &:= \langle r \rangle \trp(s) & 
        \trm(\exists r.s) &:= [r] \trm(s) \\
        \trp(\forall r.s) &:= [r] \trp(s) &
        \trm(\forall r.s) &:= \langle r \rangle \trm(s) \\
        \trp(A) &:= A &
        \trm(A) &:= \lnot A\\
        \trp(a) &:= a &
        \trm(a) &:= \lnot a\\
        \trp(\lnot s) &:= \trm(s) & 
        \trm(\lnot s) &:= \tr^+_{\pos(S),\C}(s) \\
        \trp(s) &:= \begin{cases}
            X_s &\text{if } s \in S\\
            \mu X_s.\mathit{tr}^+_{S \cup \{s\},\C}(\C(s)) &\text{if } s \not \in S
        \end{cases}&
        \trm(s) &:= \begin{cases}
            X_{\Bar{s}} &\text{if } \Bar{s} \in S\\
            \nu X_{\Bar{s}}.\mathit{tr}^-_{S \cup \{\Bar{s}\},\C}(\C(s)) &\text{if } \Bar{s} \not \in S
        \end{cases}
    \end{align*}
   \caption{Translation functions, note the negated translation of negation, $\trp(\lnot s)$: the well-founded semantics does not treat positive and negative information equally, which is captured by only keeping the positive shape names in the set $S$: $\pos(S) := S \setminus \{\Bar{s} \in S\}$.}
    \label{fig:translationfunctions}
\end{figure*}

In the following example, we illustrate that this limited alternation is needed. That is, the alternation-free $\mu$-calculus is not expressive enough to capture the well-founded semantics.

\begin{vb}\label{vb:mantasexample2}
Recall the constraints $\C = \{s \gets A \lor \exists p.\lnot t, t \gets \lnot s \lor \exists r. t\}$ discussed in Example \ref{vb:mantasexample1}. Then
\begin{align*}
    \trpe(s) &= \mu X_s. (A \lor \langle p \rangle \nu X_{\Bar{t}}.(X_s \land [r] X_{\Bar{t}})),
\end{align*}
which is a formula expressing a property not expressible in alternation free $\mu$-calculus. Note that if we evaluate $\trpe(s)$ over $\G_1$, we find $||\trpe(s)||^{\G_1}_V = \{0,\ldots,6\}$.
\end{vb} 

\begin{toappendix}

Note that the constraint definitions that are not reachable from a target are irrelevant to its validation, and hence to test $s$ we can just use $\C_s$, defined in the following way.

\newcommand{\clo}{\mathit{cl}}

\begin{definition}
    Given a constraint set $\C$, let $\clo_\C(s)$, the \textit{closure} of a shape name $s$ w.r.t. $\C$, be defined as the smallest set of shape names $s' \in \NS$ such that $s \in \clo_\C(s)$, and furthermore, if $s' \in \clo_\C(s)$ and $s' \gets \varphi \in \C$ such that $s'' \in \sub(\varphi)$, then $s'' \in \clo_\C(s)$.
    Let the \textit{restriction} of a constraint set $\C$ to a shape name $s$ be the following constraint set:
    $$
    \C_s := \{s' \gets \varphi \in \C \mid s' \in \clo_\C(s)\}.
    $$
    We call these constraints \emph{reachable from} $s$.
\end{definition}

Recall that also $s' \in \sub(\lnot s')$. It is straightforward the following must hold.
\begin{lemma}
    For each constraint set $\C$ and interpretation $\I$, we have $\I$ validates $(\C,s(a))$ iff $\I$ validates $(\C_s,s(a))$.
\end{lemma}

Without loss of generality, the proofs below assume that all constraints in $\C$ are reachable from $s$.
We first establish an assisting lemma on the structure of the constructed modal $\mu$-calculus formulas.

\begin{lemma}\label{lemma:nu_loops}
    Given $\nu X.\varphi \in \sub(\trpe(s))$ such that $X \in \sub(\varphi)$, then there is no $\mu X.\psi \in \sub(\nu X.\varphi)$ such that $X \in \sub(\psi)$.
\end{lemma}
Assuming the contrary quickly leads to the conclusion that in that case $X$ would not have been introduced; as only the positive shape names remain in the rule $\trm(\lnot s) := \tr^+_{\pos(S),\C}(s)$.

Given an interpretation $\I$, let $\atoms_\I$ be a (partial) function that maps (approximations of) modal mu-calculus formulas to sets of atoms, defined in the following way
\begin{align*}
    \atoms_{\I}&(\mu^i X_s.\varphi) := \{s(a) \mid a \in ||\mu^i X_s. \varphi||^\I_V\} \\
    &\cup \{s'(a) \mid \mu X_{s'}.\psi \in \sub(\varphi), a \in ||\mu X_{s'}. \psi||^\I_V\} \\
    &\cup \{\lnot s'(a) \mid \nu X_{\Bar{s'}}.\psi \in \sub(\varphi), a \in ||\nu X_{\Bar{s'}}. \psi||^\I_V\}.
\end{align*}

We use the notation $S_j$ as a shorthand for $W^{\uparrow j}_{\I,\C}(\emptyset)$, whenever $\I$ and $\C$ are understood from the context.

\begin{pr}\label{prop:well-founded_to_mu-calc}
    For each constraint set $\C$, each shape name $s$, each interpretation $\I$ and each natural number $i$, if $t'(b) \in W^{\uparrow i}_{\I,\C_s}(\emptyset)$, for $t' = s'$, or $t' = \lnot s'$, then there exists a $j$ such that $t'(b) \in \atoms_{\I}(\mu^j X_s.\varphi)$, for $\mu X_s.\varphi = \trpe(s)$.
\end{pr}

\begin{proof}
     Note that because we consider $\C_s$, there exists a subformula of the form $\{ \mu X_{s''}.\psi, \nu X_{\Bar{s''}}.\psi\}\cap \sub(\trpe(s)) \not = \emptyset$ for each $s'' \gets \varphi \in \C_s$. Moreover, as the function $\trpe$ is alternating between the positive (only introducing subformulas of the form $\mu X.\psi$) and negative version, $\trm$ (only introducing $\nu X.\psi$-subformulas) exactly when some $s \gets \lnot s'$ shows up, we must also conclude that if $s''(b)\in S_i$ for some $i$, there must be $\mu X_{s''}.\psi \in \sub(\trpe(s))$, and similarly with $\nu X_{\Bar{s''}}$ for $\lnot s''(b) \in \S_i$. 
     We prove this proposition by induction on $i$. For the induction basis, note $W^{\uparrow 0}_{\I,\C_s}(\emptyset) = S_0 = \emptyset$. For the induction step, assume there exists some $t(b) \in S_i \setminus S_{i-1}$. We distinguish two cases: (i) $t = s'$ and (ii) $t = \lnot s'$.
    \begin{itemize}
        \item[(i)] Let $s' \gets \varphi \in \C$ be the unique constraint with head $s'$. We distinguish the following cases.
        \begin{itemize}
            \item[-] $\varphi = A$. Note that $s'(b) \in S_i$ implies that $b \in A^\I$. Furthermore, we find there exists some $\mu X_{s'}.\psi \in \sub(\trpe(s))$. Note that $\psi = \tr^+_{S\cup s',\C_s}(A) = A$, independent of what is contained in $S$. As $||\mu X_s'.A||^\I_V = A^\I$, we find $s'(b) \in \atoms_{\I}(\mu^j X_s.\varphi)$ for each $j\geq 1$.
            \item[-] $\varphi = \lnot s''$. Note that if $s'(b) \in W^{\uparrow i}_{\I,\C_s}(\emptyset)$, then $\lnot s''(b) \in W^{\uparrow i-1}_{\I,\C_s}(\emptyset)$, thus, by using the induction hypothesis, we find that there exists a $j$ such that $\lnot s''(b) \in \atoms_{\I}(\mu^j X_s.\varphi)$. Following the definition of $\atoms_\I$, this means there must be a $\nu X_{s''}.\psi \in \sub(\trpe(s))$ such that $b \in ||\nu X_{\Bar{s''}}.\psi||^\I_{V_j}$, where $V_j$ is such that $V_j(X_s) = ||\mu^{j-1}X_s.\varphi||^\I_V$. Considering the definition of $\trpe(s)$, we distinguish two options: (a) either we also find the subformula $\mu X_{s'}.\nu X_{\Bar{s''}}.\psi \in \sub(\trpe(s))$, or (b), only the subformula $\mu X_{s'}.X_{\Bar{s''}} \in \sub(\trpe(s))$. As (b) is in contradiction with Lemma \ref{lemma:nu_loops}, we only consider case (a). If $X_{s'} \in \sub(\psi)$, this means there is some subformula $\mu X_{s'}.\chi(X_{s'}) \in \sub(\trpe(s))$ for which we can use a similar proof strategy using the least-fixed point approximation.
            If $X_{s'} \not\in \sub(\psi)$, this means that $||\mu X_{s'}.\nu X_{\Bar{s''}}.\psi||^\I_{V_j} = ||\nu X_{\Bar{s''}}.\psi||^\I_{V_j}$, and thus $s'(b) \in \atoms_{\I}(\mu^j X_s.\varphi)$.
            \item[-] Other cases are treated similarly. 
        \end{itemize}
        \item[(ii)] 
        Let $U$ be the greatest unfounded set w.r.t. $S_{i-1}$, $\I$ and $\C_s$, we distinguish four cases: (a) $s'(b) \in U$ such that $s' \gets A \in \C$ is the unique constraint with head $s'$; (b) $s'(b) \in U$ such that $s' \gets \lnot s'' \in \C$ is the unique constraint with head $s'$; (c) $s'(b) \in U$ such that $s' \gets \exists r.s'' \in \C$ is the unique constraint with head $s'$ and $\{e \in \Delta^\I \mid (b,e) \in r^\I\} = \emptyset$; and (d) all $s'(b) \in U$ that do not fit in any of the previous categories.
        Note that cases (a) and (c) can be treated similarly to the cases $\varphi = A$, resp. $\varphi = \forall r.s''$ in part (i). Thus, we will only address cases (b) and (d) here.
        \begin{itemize}
            \item[(b)] $\varphi = \lnot s''$. Note that $\lceil \varphi \rceil^\I_{S_{i-1} \cup \lnot.U} = \{a \in \Delta^\I \mid s''(a) \not \in S_{i-1} \cup \lnot.U\}$. Thus, $b \not \in \lceil \varphi \rceil^\I_{S_{i-1} \cup \lnot.U}$ implies that $s''(b) \in S_{i-1}$. Now, we may use the induction hypothesis to conclude there exists some $j$ such that $s''(b) \in \atoms_\I(\trpe(s))$. The rest of the argument can be treated similar to the cases in (i).
            \item[(d)] Take some $s'(b) \in U$ and note that there must exist some $\nu X_{\Bar{s'}}.\psi \in \sub(\trpe(s))$. Now let $L = \{a \in \Delta^\I \mid s'(a) \in U\}$, we aim to show that $L \subseteq ||\nu X_{\Bar{s'}}.\psi||^\I_{V[X_{\Bar{s'}} \mapsto L]}$. We distinguish two cases: ($\alpha$) either $X_{\Bar{s'}} \in \sub(\psi)$, or ($\beta$) not. In the latter case, we may distinguish two settings again, either there exists some $\nu X_{\Bar{s''}}.\psi' \in \sub(\trpe(s))$ such that $\nu X_{\Bar{s'}}.\psi \in \sub(\psi')$ and $X_{\Bar{s''}} \in \sub(\psi)$, that is, $\nu X_{\Bar{s'}}.\psi$ appears within some ``loop'' of negated atoms, or not. The case of $\nu X_{\Bar{s'}}.\psi$ in the ``loop'' will follow from the way case $(\alpha)$ is considered, whereas the outside-of-the-loop-case will either fall directly in cases (a)-(c), or can be brought back to cases of the form (a)-(c) with techniques similar to the ones in (i). 
            
            That is, we solely focus on case $(\alpha)$ here; if $X_{\Bar{s'}} \in \sub(\psi)$, there exist one or more chains of axioms $C_i = \{s_{i_1}\gets \varphi_{i_1},\ldots,s_{i_n} \gets \varphi_{i_n}\}$ such that $s_{i_1} = s'$, each $\varphi_{i_j}$ is of one of the forms $s'' \land s''', s'' \lor s''', \exists r.s''$ or $\forall r.s''$, for each $1 \leq j \leq n-1$ we have $s_{i_{j+1}} \in \sub(\varphi_{i_j})$, and lastly, we find $s_{i_1} \in \sub(\varphi_{i_n})$. To avoid redundancy, we also assume $s' \not = s_{i_j}$ for each $1 < j \leq n$. Note that by combining the insights of Lemma \ref{lemma:nu_loops} with the definition of the translation function, we find that for each $1 <j \leq n$, there exists some $\nu X_{\Bar{s_{i_j}}}.\psi_{i_j} \in \sub(\psi)$ such that $\nu X_{\Bar{s_{i_j}}}.\psi_{i_j} \in \sub(\psi_{i_{j-1}})$.

            Now let $U'$ be the greatest subset of $U$ such that if $s(a) \in U'$, then there exists some chain of axioms $C_i$ such that $s = s_{i_j}$ for some $1 \leq j \leq n$, and moreover, for all $1 \leq j \leq n$, if $s_{i_j}(a) \in U'$ and $s_{i_j} \gets \varphi_{i_j} \in \C$, then
            \begin{itemize}
                \item if $\varphi_{i_j} \in \{s_{i_{j+1}} \land s', s' \land s_{i_{j+1}}, s_{i_{j+1}} \lor s', s' \lor s_{i_{j+1}}\}$, then $s_{i_{j+1}}(a) \in U'$; and
                \item if $\varphi_{i_j} \in \{\exists r.s_{i_{j+1}}, \forall r.s_{i_{j+1}}\}$, then there exists some $b \in \{b \in \Delta^\I \mid (a,b) \in r^\I\}$ such that $s_{i_{j+1}}(b) \in U'$,
            \end{itemize}
            where $s_{i_{n+1}}$ is used as a substitute name for $s_{i_1}$.
            
            For each chain of axioms $C_i$, let $L_{i_j} = \{a \in \Delta^\I \mid s_{i_j}(a) \in U'\}$. As $s_{i_1} = s'$ for each chain of axioms $C_i$, we find $L_{i_1} = L_{k_1}$ for each pair of axiom chains $C_i$ and $C_k$. We want to show that 
            \begin{align}\label{loop_greatest_unfounded_set}
                L_{i_j} \subseteq ||\nu X_{\Bar{s_{i_j}}}.\psi_{i_j}||^\I_{V[X_{\Bar{s_{i_1}}} \mapsto L_{i_1}]}
            \end{align}
            for each constraint $s_{i_j} \gets \varphi_{i_j}$ appearing in one of the axiom chains $C_i$, assuming that for each $s_{i_j} \gets \varphi$ with $\varphi \in \{s_{i_{j+1}} \lor s', s' \lor s_{i_{j+1}}\}$, $\{a \in \Delta^\I \mid s_{i_j}(a) \in U'\} \subseteq ||\nu X_{\Bar{s'}}.\psi'||^\I_V$, where $V(X_{\Bar{s_{k_1}}}) = L_{k_1}$ in case $s' = s_{k_l}$, that is, $s'$ appears in some other chain, or where $V$ is in the state it is in after approximating the previous layer in case $s'$ does not appear in any other chain.
            To see (\ref{loop_greatest_unfounded_set}) holds, note it is immediate that $L_{i_1} \subseteq ||X_{\Bar{s_{i_1}}}||^\I_{V[X_{\Bar{s_{i_1}}} \mapsto L_{i_1}]}$. 
            Now assume $L_{i_j} \subseteq ||\nu X_{\Bar{s_{i_j}}}.\psi_{i_j}||^\I_{V[X_{\Bar{s_{i_1}}} \mapsto L_{i_1}]}$, and let $s_{i_{j-1}}(b) \in L_{i_{j-1}}$. We distinguish four cases: conjunction, disjunction, and existential and universal quantification. To illustrate the main idea, we treat the conjunction case.
            \begin{itemize}
                \item Let $s_{i_{j-1}} \gets s_{i_{j}} \land s' \in C_i$ (or similarly $s_{i_{j-1}} \gets s' \land s_{i_{j}} \in C_i$). In this case $\psi_{i_{j-1}} = \nu X_{\Bar{s_{i_{j}}}}.\psi_{i_{j}} \lor \chi$, for some $\chi$ representing the $\nu$ formula for $s'$ not included in the discussed loop, and thus $||\nu X_{\Bar{s_{i_j}}}.\psi_{i_j}||^\I_{V[X_{\Bar{s_{i_1}}} \mapsto L_{i_1}]} \cup ||\chi||^\I_{V[X_{\Bar{s_{i_1}}} \mapsto L_{i_1}]}= ||\psi_{i_{j-1}}||^\I_{V[X_{\Bar{s_{i_1}}} \mapsto L_{i_1}]}$. As $X_{\Bar{s_{i_{j-1}}}} \not \in \sub(\psi_{i_{j-1}})$, we find $||\psi_{i_{j-1}}||^\I_{V[X_{\Bar{s_{i_1}}} \mapsto L_{i_1}]} = ||\nu X_{\Bar{s_{i_{j-1}}}}.\psi_{i_{j-1}}||^\I_{V[X_{\Bar{s_{i_1}}} \mapsto L_{i_1}]}$.
                Moreover, we find for each $a \in L_{i_{j-1}}$, we have $a \in L_{i_j}$, or $a \in ||\chi||^\I_{V[X_{\Bar{s_{i_1}}} \mapsto L_{i_1}]}$; suppose $a \in L_{i_{j-1}}$, then $a \not \in \mtrue{s_{i_j} \land s'}{\G}{S \cup \lnot.U'}$, which means $a\not \in \mtrue{s_{i_j}}{\G}{S \cup \lnot.U'}$ or $\mtrue{s_{i_j} \land s'}{\G}{S \cup \lnot.U'}$. In the first case, note that $\mtrue{s_{i_j}}{\G}{S \cup \lnot.U'} = \{c \in \Delta^\G \mid \lnot s_{i_j} \not \in S \cup \lnot.U'\}$, thus $s_{i_j}(a) \not \in U'$, which means that $a \in L_{i_j}$. In the second case, $a \not \in \mtrue{s'}{\G}{S \cup \lnot.U'}$, there are again two options: either $s'=s_{k_l}$, in which case a similar argument for the loop $k$ suffices to show $a \in L_{k_l} \subseteq ||\chi||^\I_{V[X_{\Bar{s_{i_1}}} \mapsto L_{i_1}]}$; or $s'$ does not appear in any loop. In particular, this means there are no atoms of the form $s'(b) \in U'$, by definition of $U'$. If we now look at $a \not \in \mtrue{s'}{\G}{S \cup \lnot.U'}$, this implies $s'(a) \in U'$, or $\lnot s'(a) \in S$. Thus, $\lnot s'(a) \in S$. So, we are back in a non-loop case and can use the reasoning of before to conclude: $a \in ||\chi||^\I_{V[X_{\Bar{s_{i_1}}} \mapsto L_{i_1}]}$. Combining all subsumptions and equalities leads to the conclusion that $L_{i_{j-1}} \subseteq ||\nu X_{\Bar{s_{i_{j-1}}}}.\psi_{i_{j-1}}||^\I_{V[X_{\Bar{s_{i_1}}} \mapsto L_{i_1}]}$, as required.
            \end{itemize}
        \end{itemize}
        
    \end{itemize}

\end{proof}

\begin{pr}\label{prop:mu-calc_to_well-founded}
    For each constraint set $\C$, each interpretation $\I$ and each natural number $i$, such that $\mu X_s.\varphi = \trpe(s)$ if $s'(b) \in \atoms_{\I}(\mu^i X_s.\varphi)$, then there exists $j$ such that $s'(b) \in W^{\uparrow j}_{\I,\C}(\emptyset)$.
\end{pr}

\begin{proof}
As we are considering a least fixed point, it suffices to prove the following claim.

    \begin{cl}
        Given some $\sigma X_t.\varphi$ and a valuation function $V$, and assume that for each variable $X_{t'} \in \sub(\varphi)$ that is free in $\sigma X_t.\varphi$, we have, if $b \in V(X_{t'})$, then there exists a $j$ such that $b \in \lfloor t' \rfloor^\I_{S_j}$. 
        In that case, there exists a $j'$ such that $||\sigma X_t.\varphi||^\I_V \subseteq \lfloor t \rfloor^\I_{S_{j'}}$.
    \end{cl}

    We prove the above claim by induction on the construction of $\varphi$. We distinguish three cases: (i) $\sigma X.\varphi$ such that $X \not \in \sub(\varphi)$, and in case $X \in \sub(\varphi)$: (ii) $\sigma = \mu$ and (iii) $\sigma = \nu$. 

    \begin{enumerate}
        \item[(i)] \begin{itemize}
        \item[-] $\varphi = A$. Clearly $||\sigma X_t.A||^\I_V = ||A||^\I_V = A^\I = \lfloor A \rfloor^\I_{S_j}$, as $\sigma X_t.A$ can only be constructed in case $t \gets A \in \C$, we find that if $b \in ||\sigma X_t.A||^\I_V$, then $b \in \lfloor t \rfloor^\I_{S_{j+1}}$. 
        \item[-] $\sigma X_t.\varphi = \mu X_t.\nu X_{\Bar{s}}.\psi$. As $X_t \not \in \sub(\psi)$, we find $||\mu X_t.\nu X_{\Bar{s}}.\psi||^\I_V = ||\nu X_{\Bar{s}}.\psi||^\I_V$, then by IH, there exists some $j$ such that $||\nu X_{\Bar{s}}.\psi||^\I_V \subseteq \lfloor \lnot s \rfloor^\I_{S_j}$. As $\mu X_t.\nu X_{\Bar{s}}.\psi$ will only be constructed given some $t \gets \lnot s \in \C$, we find that if $b \in ||\mu X_t.\nu X_{\Bar{s}}.\psi||^\I_V$, then $b \in \lfloor t \rfloor^\I_{S_{j+1}}$.
        \item[-] Other cases are treated similarly.
    \end{itemize}
    \item[(ii)]
    \begin{itemize}
        \item[-] $\mu X_t.\varphi(X_t) = \mu X_t.\sigma X_s.\psi(X_t)$. In this case, we find $||\mu X_t.\sigma X_s.\psi(X_t)||^\I_V$ can be determined by using the approximation function $||\mu^{0} X_t.\sigma X_s.\psi(X_t)||^\I_V = \emptyset$ and $||\mu^{i+1} X_t.\sigma X_s.\psi(X_t)||^\I_V = ||\sigma X_s.\psi(X_t)||^\I_{V[X_t \mapsto ||\mu^{i} X_t.\sigma X_s.\psi(X_t)||^\I_V]}.$ Note that for the $i$-th approximated case, we can derive that there exists some $j$ such that $||\mu^{i} X_t.\sigma X_s.\psi(X_t)||^\I_V \subseteq \lfloor t \rfloor^\I_{S_j}$ by a repetitive usage of the induction hypothesis. As we are considering a least fixed point, this eventually suffices.
        \item[-] Other cases are treated similarly.
    \end{itemize}
    \item[(iii)] $\nu X_t.\varphi(X_t)$. Following Lemma \ref{lemma:nu_loops}, all $\sigma X_{t'}.\psi \in \sub(\varphi)$ such that $X_t \in \sub(\psi)$ are such that $\sigma = \nu$. We collect all relevant shape names in the following set: let $W = \{s \in \NS \mid \nu X_{\Bar{s}}.\psi \in \sub(\varphi), X_t \in \sub(\psi) \}$. Note that for all $\nu X_{\Bar{s}}.\psi \in \sub(\varphi)$ such that $X_t \not \in \sub(\psi)$ the induction hypothesis may be applied to derive that there exists some $j$ such that $||\nu X_{\Bar{s}}.\psi ||^\I_V \subseteq \lfloor \lnot s \rfloor^\I_{S_j}$. Let $j'$ be the maximum of all these $j$. Now let 
        $U = \{s(a) \mid a \in ||\nu X_{\Bar{s}}.\psi||^\I_V, \nu X_{\Bar{s}}.\psi \in \sub(\nu X_t.\varphi), s \in W\}.$
        
        In what follows, we will show $U$ is an unfounded set w.r.t.\ $S_{j'}$, $\I$ and $\C$. 
        To this end, we have to show that for all $s(a) \in U$ we find $a \not\in \lceil \C(s) \rceil^\I_{S_{j'} \cup \lnot.U}$. 
        Note that if $s \gets \varphi' \in \C$ such that $s \in W$, then $\varphi'$ is of the form $\exists r.s'$, $\forall r.s'$, $s \gets s' \land s''$ or $s \gets s' \lor s''$. Thus, we may consider $\C(s) \in \{\exists r.s',\forall r.s', s' \land s'', s' \lor s''\}$:
    \begin{itemize}
        \item[-] $\C(s) = \forall r.s'$. In this case, the subformula looks like $\nu X_{\Bar{s}}.\langle r \rangle \nu X_{\Bar{s'}}.\psi'$, or $\nu X_{\Bar{s}}.\langle r \rangle X_{s'}$. Note that in the latter case, $s' = t$. As we assumed $a \in ||\nu X_{\Bar{s}}.\psi||^\I_V$, we find there must be a $b \in \Delta^\I$ such that $(a,b) \in r^\I$ and $b \in ||\nu X_{\Bar{s'}}.\psi''||^\I_V$, both if $\psi'' = \psi'$ and if $\psi'' = \varphi$. Thus, $s'(b) \in U$, which means $s'(b) \not \in \lnot.U$ and moreover, we find $s'(b) \not \in S_{j'}$ (I.H. of the previous proposition), thus $b \not\in \lceil s' \rceil^\I_{S_{j'}\cup \lnot.U}$. As $\lceil \forall r.s' \rceil^\I_{S_{j'}} = \{a' \in \Delta^\I \mid \forall a' \in \Delta^\I : (a,a')\in r^\I \text{ implies } a'\in\lceil s' \rceil^\I_{S_{j'}\cup \lnot.U}\}$, we must conclude that, because of the witness $b$, $a \not \in \lceil \C(s) \rceil^\I_{S_{j'} \cup \lnot.U}$.
        \item[-] $\C(s) = s' \lor s''$. In this case, the subformula looks like $\nu X_{\Bar{s}}.\psi' \land \psi''$, for $\psi'$ and $\psi''$ either a variable $Y$ or a subformula of the form $\nu Y.\chi$. There are two possible options: either $\{s',s''\} \subseteq W$, in which case a similar argument as before suffices. However, if, without loss of generality, $s' \not \in X$, the subformula looks like $\nu X_{\Bar{s}}.(\nu X_{\Bar{s'}}.\psi') \land \psi''$
        we are back in the case of $\nu X_{\Bar{s'}}.\psi \in \sub(\varphi)$ such that $X_t \not \in \sub(\psi)$, and thus $||\nu X_{\Bar{s'}}.\psi||^\I_V \subseteq \lfloor \lnot s' \rfloor^\I_{S_{j'}}$, as concluded before. As we assumed $a \in ||\nu X_{\Bar{s}}.\psi||^\I_V$, this means that also $a \in ||\nu X_{\Bar{s'}}.\psi||^\I_V$, and thus $a \in \lfloor \lnot s' \rfloor^\I_{S_{j'}}$. This corresponds to $\lnot s'(a) \in S_{j'}$, so clearly also $a \not \in \lceil s' \rceil^\I_{S_{j'} \cup \lnot.U}$. To argue why $a \not \in \lceil s'' \rceil^\I_{S_{j'} \cup \lnot.U}$, a similar argument as above applies, as $s'' \in W$. Thus, we find $a \not \in \lceil \C(s) \rceil^\I_{S_{j'} \cup \lnot.U}$.
        \item[-] Other cases are treated similarly. 
    \end{itemize}
    \end{enumerate}
\end{proof}
\end{toappendix}

We can now prove the correctness of our translation. 

\begin{theoremrep}\label{theorem:translationiscorrect}
    For each shapes pair $(\C,s)$, each data graph $\I$, and each node $c$ of $\I$, we have $$\I \text{ validates } (\C,s(c)) \quad\text{iff}\quad \I,c \models \trpe(s).$$
\end{theoremrep}

The proof, given in the appendix, uses the approximation semantics of the $\mu$-calculus, and the iterative computation of the well-founded semantics via the immediate consequence operator. 
In a nutshell, we show that for every $a$ and $i$:
\begin{enumerate}
    \item if $s(a)$ is in the $i$-th iteration of $W_{\I,\C_s}(\emptyset)$, then there exists a $j$ such that $a$ is in the extension of the $j$-th approximation of $\mu X_s.\varphi = \trpe(s)$, and conversely 
    \item if $a$ is in the extension of the $j$-th approximation of $\mu X_s.\varphi = \trpe(s)$, then for some $j$ we have that 
    $s(a)$ is in the $i$-th iteration of $W_{\I,\C_s}(\emptyset)$.
\end{enumerate}
Each of these claims is shown by induction on the shape of the formula, which means we also treat atoms of the form $s'(a)$ and $\lnot s'(a)$ in the $i$-th iteration of $W_{\I,\C_s}(\emptyset)$ and compare this with the extension of the subformulas $\mu X_{s'}.\psi$, resp. $\nu X_{\Bar{s'}}.\psi$ in the $j$-approximation of $\mu X_s.\varphi = \trpe(s)$.
Note that the iterative computation of $W_{\I,\C}$ and the interactive approximations of $\trpe(s)$ operate quite differently and, in particular, they do not derive facts in the same number of iterations, hence the mismatch between the indices $i$ and $j$.

This theorem gives us desired reduction from shape satisfiability to satisfiability in the full hybrid $\mu$-calculus.

\begin{corollary}\label{cor:shapeSat}
    Let $\C_1$ and $\C_2$ be sets of shape constraints and $s_1$ and $s_2$
    shape names. Then 
    \begin{itemize}
        \item $s$ is (finitely) satisfiable w.r.t.\ $\C_1$ under the well-founded semantics iff   $\mathit{tr}^+_{\emptyset,\C_1}(s)$ is (finitely) satisfiable, and
        \item $(\C_1,s_1)$ (finitely) implies $(\C_2,s_2)$ under the well-founded semantics iff   $\mathit{tr}^+_{\emptyset,\C_1}(s) \land \lnot \mathit{tr}^+_{\emptyset,\C_2}$ is not (finitely) satisfiable.
    \end{itemize}
\end{corollary}

\begin{proof}
    First, note that for each $\I$ there exists an $i \in \mathbb{N}$ such that the least fixed point coincides with its approximation: $||\mu X. \varphi||^\I_V = ||\mu^i X. \varphi||^\I_V$. Similarly, there exists an $i' \in \mathbb{N}$ such that $W^{\uparrow i'}_{\I,\C}(\emptyset) = W^{\uparrow j}_{\I,\C}(\emptyset)$ for all $j > i'$.
    Propositions \ref{prop:well-founded_to_mu-calc} and \ref{prop:mu-calc_to_well-founded} apply specifically to these $i$ and $i'$, which suffices to conclude the above. 
\end{proof}

\paragraph{Implication and Satisfiability of Documents.}

We also reduce satisfiability and containment of SHACL documents to satisfiability in the full hybrid $\mu$-calculus, as presented below.  
Here, we show how to deal with the (finite) document implication problem, and recall that in Proposition \ref{pr:docsat-docimp} we described how (finite) document satisfiability can be reduced to this case. 

That is, for any document $(\C,\T)$, let   $\Theta_{\C,\T}$ be the formula defined as a conjunction of the following formulas:
\begin{align*}
    &\neg a \lor \trpe(s) &&\text{for all }(a,s)\in \T\\
    &\neg A \lor \trpe(s) &&\text{for all }(A,s)\in \T\\
    &\neg\langle r \rangle\top \lor \trpe(s) &&\text{for all }(r,s)\in \T.
\end{align*}

When evaluated on a graph node, $\Theta_{\C,\T}$ tells us whether the node satisfies all required target shapes. To propagate $\Theta_{\C,\T}$ to all (relevant) nodes of the graph, we use a greatest fixed point formula; for a document $(\C,\T)$, and a role name $p$, let 
\[\Lambda_{\C,\T,p}=  \nu X.(\Theta_{\C,\T} \land \bigwedge_{r\in R \cup \{p\}}[ r^-
].X \land \bigwedge_{r\in R \cup \{p\}}[ r
].X ) , \]
where $R$ is the set of role names that  appear in $(\C,\T)$.  

\begin{theorem}\label{thm:docImplicReduc}
For each pair of SHACL documents $(\C,\T)$ and $(\C',\T')$, a SHACL document
$(\C,\T)$ (finitely) implies $(\C',\T')$  iff 
$$
\bigwedge_{a\in I}\langle p \rangle (a \land \Lambda_{\C,\T,p}) \land \langle p\rangle \neg \Lambda_{\C',\T',p}
$$ 
is not (finitely) satisfiable, where $I$ is the set of individuals that appear in $(\C,\T)$ or $(\C',\T')$, and $p$ a fresh role.
\end{theorem}

\section{Complexity Results}

Since the full hybrid $\mu$-calculus is decidable, we obtain from Corollary~\ref{cor:shapeSat}  and Theorem~\ref{thm:docImplicReduc} that the static analysis problems of interest are also decidable. 
This is particularly exciting in light of Theorem \ref{thm:undecSupp}. 
However, we also want to obtain tight complexity results. 
The complexity of the $\mu$-calculus and its variants is well understood, but, regrettably, our translation may incur an exponential blow-up due to the possible repetition of subformulas associated with each shape name.
Fortunately, this blow-up affects only the representation of the translation as a single formula, and not the complexity of the problem. 

In this section, we obtain complexity results for the static analysis problems under the well-founded semantics.
The key insight is that the automata-based technique for the fully hybrid $\mu$-calculus by \citeauthor{DBLP:conf/cade/SattlerV01} (\citeyear{DBLP:conf/cade/SattlerV01}) can be adapted to decide satisfiability of the formulas that result from our translation in time that is single exponential in the size of the original input, and not of the translated formulas.
In a nutshell, 
instead of using all subformulas as states, we reuse the translation functions as presented in Figure \ref{fig:translationfunctions} as the transition function, avoiding the problem of exponentially many states in the constructed automaton. 
We provide the full construction for the sake of completeness, but remark that the adaptation is relatively straightforward, as will be clear to the readers familiar with automata decision procedures for the $\mu$-calculus.

\paragraph{Two-Way Alternating Parity Tree Automata.}
We recall 2ATA, which can move up and down possibly infinite trees with transitions that may move to combinations of states and directions.
For $k \geq 1$, let $(\{1,\ldots,k\}^*,\ell)$ be a $k$-ary $\Sigma$-\textit{labelled tree} if $\ell$ maps each node $x \in \{1,\ldots,k\}^*$ to its label $\ell(x) \in \Sigma$. For $1 \leq i \leq k$, we let $x \cdot i$ be the $i$th \textit{child} of $x$. With $x \cdot 0$ we mean $x$ itself, and $x \cdot -1$ indicates the \textit{parent} of $x$; the node $y$ such that $x = y \cdot i$ for some $1 \leq i \leq k$.
Let $\mathcal{B}(I)$ be the set of \textit{positive boolean formulas} over $I$, built inductively by applying $\land$ and $\lor$ starting from $\bot$, $\top$ and the elements of $I$. For a set $J \subseteq I$ and a formula $\varphi \in \mathcal{B}(I)$, we say that $J$ satisfies $\varphi$ iff assigning $\top$ to the elements in $J$ and $\bot$ to those in $I \setminus J$, makes $\varphi$ true. For a positive integer $k$, let $[k] = \{-1, 0,1, \dots,k \}$. 
A \textit{two-way alternating parity tree automaton} (2ATA) running over infinite tree encodings with branching degree $k$, is a tuple $\automaton = \langle \Sigma, Q, \delta, q_0, p \rangle$, where $\Sigma$ is the input alphabet, $Q$ is a finite set of states, $\delta: Q \times \Sigma \rightarrow \mathcal{B}([k]\times Q)$ is the transition function, $q_0 \in Q$ is the initial state and $p: Q \rightarrow \mathbb{N}$ is the priority function. 

A run of a 2ATA $\automaton$ over a $k$-ary tree $(\{1,\ldots,k\}^*,\ell)$, is a $(\{1,\ldots,k\}^* \times Q)$-labelled tree $\T_r=(T_r, \ell_r)$ with $V_r \subseteq \{1,\ldots,k\}^*$, satisfying: 
\begin{enumerate}
    \item $\epsilon \in V_r$ and $\ell_r(\epsilon) = (\epsilon,q_0)$. 
    \item Let $y \in V_r$ with $\ell_r(y) = (x,q)$ and $\delta(q,\ell(x)) = \varphi$. Then there is a (possibly empty) set $S = \{ (c_1,q_1), \dots, (c_n,q_n) \} \subseteq [k] \times Q$ such that: 
    \begin{itemize}
        \item $S$ satisfies $\varphi$, and 
        \item for all $1 \leq i \leq k$, we have that $y \cdot i \in V_r$ and $\ell_r(y  \cdot i) = (x \cdot c_i, q_i)$.
    \end{itemize}
\end{enumerate}
Given an infinite path $p \subseteq V_r$, let $\mathit{inf}(p) \subseteq Q$ be the set of states that appear infinitely often in $p$. A run $\T_r$ of a 2ATA $\automaton = \langle \Sigma, Q, \delta, q_0, p \rangle$ is \emph{accepting} if for every infinite path $p \subseteq V_r$, the state with the highest priority that occurs in $\mathit{inf}(p)$ is even.

Following \citeauthor{DBLP:conf/cade/SattlerV01} (\citeyear{DBLP:conf/cade/SattlerV01}), we define \emph{guesses}, which fix the entire configuration of the individuals mentioned in the constraint set. We can enumerate all such configurations and build an automaton for each of them.

\begin{definition}
    A \emph{guess} $\guess = (G,f,C)$ for a constraint set $\C$ containing the nominals $a_1,\ldots,a_l$ consists of a \emph{guess list} $G = (\gamma_1,\ldots,\gamma_l)$, a set of \emph{connections} $C \subseteq \NI \times \NR(\C) \times \NI$ and a \emph{guess mapping} $f : \{1,\ldots,l\} \rightarrow \{1,\ldots,l\}$, such that for each $1 \leq i,j \leq l$, $\emptyset \subsetneq \gamma_i \subseteq \sub(\C) \cup \{\lnot s \mid s \in \sub(\C)\}$ or $\gamma_i = \bot$, $a_i \in \gamma_{f(i)}$, $a_i \not \in \gamma_j$, for all $j \not = f(i)$, if $\NI \cap \gamma_i = \emptyset$, then $\gamma_i = \bot$, and $(a_i,r,a_j) \in C$ implies $(a_j, r^-,a_i) \in C$.
\end{definition}

\begin{prrep}\label{pr:2ata}
    For each shapes pair $(\C,s)$ and each guess $\guess$ for $\C$,
    we can define a 2ATA $\automaton(\C,s,\guess)$ such that:
    \begin{enumerate}
        \item If $\trpe(s)$ is satisfiable, then for some guess $\guess$, the language recognised by $\automaton(\C,s,\guess)$ is non-empty.
        \item If for some guess $\guess$ the language of $\automaton(\C,s,\guess)$ is not empty, then  $\trpe(s)$ is satisfiable.
        \item The number of states in $\automaton(\C,s,\guess)$ is linear in $|\sub(\C)|$.
    \end{enumerate}
\end{prrep}

The construction can be found in the appendix. It is a relatively straightforward adaptation of the translation in \citeauthor{DBLP:conf/cade/SattlerV01} (\citeyear{DBLP:conf/cade/SattlerV01}). The core of the construction uses two states 
$\trpt(\varphi)$ and $\trmt(\varphi)$ for each $\varphi$ in $\sub(\C)$, corresponding to the two functions of the translation. Note that we can ignore the set $S$ in this setting, as we no longer try to build a linear formula capturing a form of circularity. Moreover, the parity condition setting $p(\trpt(\varphi)) = 1$ for each $\varphi$ in $\sub(\C)$, and for all other states $p(q) = 0$, reflects exactly the limited alternation of our translation.

\begin{proof}
    We use $a_1,\ldots,a_l$ to denote the nominals occurring in $\C$, 
    Furthermore, we assume all input trees are $k$-ary full trees; all non-leaf nodes have $k$ children, where $k$ is the maximum of the amount of $\langle r \rangle \psi \in \sub(\C)$, and $l$ plus one. 

    Let $\automaton(\C,s,\guess) := \langle \Sigma, Q, \delta, \Tilde{q}, p\rangle$, be defined as follows. As the automaton construction remains so close to the automaton constructed by Sattler and Vardi (2001), we will just provide the technical details and refer the reader to their work for more intuition and explanation. Let
    $$
    \Sigma = \{\bot, \rot\} \cup \{\sigma \mid \sigma \subseteq \NC(\C) \cup \NI(\C) \cup \roles(\C) \cup \{\xrightarrow{r}a_i \mid r\in \roles(\C), 1 \leq i \leq l\}\},
    $$
    where $\NC(\C)$, $\NI(\C)$ and $\roles(\C)$ are the restrictions of the alphabets $\NC$, $\NI$ and $\roles$ to the symbols used in $\C$.
    Let
    \begin{align*}
        Q &= \{\bot,\Tilde{q},q_0,q_0',q_1,\ldots,q_l,q,q',q''\} \cup \{\trpt(\varphi), \trmt(\varphi) \mid \varphi \in \sub(\C) \cup \{\lnot s \mid s \in \sub(\C)\}\} \cup \{r,\lnot r \mid r \in \roles(\C)\}\\
        p(q) &= \begin{cases}
            1  &\text{if } q = \trpt(\varphi), \varphi \in \sub(\C)\\
            0 &\text{otherwise.}
        \end{cases}
    \end{align*}
Note it is immediate that the number of states in linear in $|\sub(\C)|$. Finally, we define the transition function in the following way, for each $\sigma \in \Sigma$.
\begin{align*}
    \delta(\Tilde{q},\sigma) &= (0,q_0) \land (0,q_0')\\
    \delta(q_0,\sigma) &= \begin{cases}
        \bigwedge_{i=1}^l (i,q_i) \land \bigwedge_{i=l+1}^k (i,q) \land \bigwedge_{i=1}^k (i,q'') &\text{if } \rot=\sigma\\
        \bot &\text{otherwise}
    \end{cases}\\
    \delta(q'',\sigma) &= \begin{cases}
        \top &\text{if } r \not \in \sigma, r^- \not \in \sigma, \text{ for each } r \in \roles(\C)\\
        \bot &\text{otherwise}
    \end{cases}
\end{align*}
In the following, let $1 \leq i \leq l$, $\hat{q} \in Q \setminus \{q_0,q_1,\ldots,q_l,q'',q\}$, $r \in \roles(\C)$, 
$p \in \NI \cup NC$, $\hat{\sigma} \in \Sigma$, $\sigma \in \Sigma \setminus \{\bot\}$ and moreover
\begin{align*}
    \Gamma(i) &= \begin{cases}
        (i,\bot) &\text{if } \gamma_i = \bot\\
        \bigwedge_{\varphi \in \gamma_i} (i,\trpt(\varphi)) &\text{if } \gamma_i \subseteq \sub(\C) \cup \{\lnot s \mid s \in \sub(\C)\}
    \end{cases}\\
\end{align*}
then
\begin{align*}
    \delta(q_i,\hat{\sigma}) &= \begin{cases}
        \bigwedge_{i=1}^k (i,q) &\text{if } \gamma_i \cap (\NI \cup \NC) = \hat{\sigma} \cap (\NI \cup \NC), \rot \not = \hat{\sigma}, \\
        &\text{and, for each } a \in \NI \cap \hat{\sigma} \text{ and } (a,r,a') \in C, \xrightarrow{r} a' \in \hat{\sigma}\\
        \bot &\text{otherwise}
    \end{cases}\\
    \delta(q,\hat{\sigma}) &= \begin{cases}
        \bigwedge_{i=1}^k (i,q) &\text{if } \hat{\sigma} \cap \NI = \emptyset \text{ and } \rot \not = \hat{\sigma}\\
        \bot &\text{otherwise}
    \end{cases}\\
    \delta(\hat{q},\bot) &= \begin{cases}
        \top &\text{if } \hat{q} = \bot\\
        \bot &\text{otherwise}
    \end{cases}\\
    \delta(\bot,\hat{\sigma}) &= \begin{cases}
        \top &\text{if } \hat{\sigma} = \bot\\
        \bot &\text{otherwise}
    \end{cases}\\
    \delta(q_0',\sigma) &= \bigwedge_{i = 1}^l \Gamma(i) \land \bigvee_{i=1}^k (i,\trpt(s)) \land \land \bigwedge_{i=1}^k ((1,q') \lor (i,\bot))\\
    \delta(q',\sigma) &= \bigwedge_{\xrightarrow{r^-}a_i \in \sigma, \forall r.s \in \gamma_{f(i)}} (0,\trpt(s)) \land \bigwedge_{\xrightarrow{r^-}a_i \in \sigma, \{\exists r.s, \lnot s'\} \subseteq \gamma_{f(i)},s' \gets \exists r.s \in \C} (0,\trmt(s)) \land \bigwedge_{i = 1}^k ((i,q') \lor (i,\bot))\\
    \delta(r,\sigma) &= \begin{cases}
        \top &\text{if } r \in \sigma\\
        \bot &\text{otherwise}
    \end{cases}\\
    \delta(\lnot r,\sigma) &= \begin{cases}
        \top &\text{if } r \not \in \sigma \text{ and } \sigma \not = \rot\\
        \bot &\text{otherwise}
    \end{cases}\\
    \delta(\trpt(p),\sigma) &= \begin{cases}
        \top &\text{if } p \in \sigma\\
        \bot &\text{otherwise}
    \end{cases}\\
    \delta(\trpt(s \land s'),\sigma) &= (0,\trpt(s)) \land (0,\trpt(s'))\\
    \delta(\trpt(s \lor s'), \sigma) &= (0,\trpt(s)) \lor (0,\trpt(s'))\\
    \delta(\trpt(\lnot s),\sigma) &= (0,\trmt(s))\\
    \delta(\trpt(s),\sigma) &= (0,\trpt(\C(s)))\\
    \delta(\trpt(\exists r.s),\sigma) &= \begin{cases}
        \top &\text{if } \xrightarrow{r}a_i \in \sigma, s \in \gamma_{f(i)}\\
        \bigvee_{j=1}^k ((j,\trpt(s)) \land (j,r)) &\text{otherwise}
    \end{cases}\\
    \delta(\trpt(\forall r.s),\sigma) &= \begin{cases}
        \bot &\text{if } \xrightarrow{r}a_i \in \sigma, s \not \in \gamma_{f(i)}\\
        ((-1,\trpt(s)) \lor (0,\lnot r^-)) \land \bigwedge_{j=1}^k ((j,\trpt(s)) \lor (j,\lnot r) \lor (j, \bot)) &\text{otherwise}
    \end{cases}\\
    \delta(\trmt(p),\sigma) &= \begin{cases}
        \top &\text{if } p \not \in \sigma \text{ and } \sigma \not = \rot\\
        \bot &\text{otherwise}
    \end{cases}\\
    \delta(\trmt(s \land s'),\sigma) &= (0,\trmt(s)) \lor (0,\trmt(s'))\\
    \delta(\trmt(s \lor s'), \sigma) &= (0,\trmt(s)) \land (0,\trmt(s'))\\
    \delta(\trmt(\lnot s),\sigma) &= (0,\trpt(s))\\
    \delta(\trmt(s),\sigma) &= (0,\trmt(\C(s)))\\
    \delta(\trmt(\exists r.s),\sigma) &= \begin{cases}
        \bot &\text{if } \xrightarrow{r}a_i \in \sigma, \lnot s \not \in \gamma_{f(i)}\\
        ((-1,\trmt(s)) \lor (0,\lnot r^-)) \land \bigwedge_{j=1}^k ((j,\trmt(s)) \lor (j,\lnot r) \lor (j, \bot)) &\text{otherwise}
    \end{cases}\\
    \delta(\trmt(\forall r.s),\sigma) &= \begin{cases}
        \top &\text{if } \xrightarrow{r}a_i \in \sigma, \lnot s \in \gamma_{f(i)}\\
        \bigvee_{j=1}^k ((j,\trmt(s)) \land (j,r)) &\text{otherwise.}
    \end{cases}
\end{align*}
\end{proof}

\paragraph{(Finite) Shape Satisfiability w.r.t.\ Constraints.}

Based on the aforementioned 2ATA construction, we can now derive the first set of complexity results.

\begin{theorem}\label{theorem:satisfiability-alcio}
    Deciding shape satisfiability w.r.t constraints for \alcio SHACL is \exptime-complete, under the well-founded semantics.
\end{theorem}

\begin{proof}
    As noted by \citeauthor{DBLP:conf/cade/SattlerV01} (\citeyear{DBLP:conf/cade/SattlerV01}), emptiness of $\automaton(\C,s,\guess)$ can be decided in exponential time in the number of states. Moreover, the amount of guesses is also bounded by the amount of connections and guess lists, which means there are at most exponentially many automata $\automaton(\C,s,\guess)$ that need to be tested for emptiness, thus shape satisfiability w.r.t. constraints can be decided in \exptime.
    
    \exptime-hardness follows from the fact that plain shape implication is equivalent to concept satisfiability in the description logic $\mathcal{ALCIO}$, which is known to be \exptime-complete \cite{DBLP:journals/jair/Tobies00}. 
    For a more detailed coverage of restricted settings that are already \exptime-hard, see \cite{dl2025}.
\end{proof}

For the finite version of the above problem we do not provide tight bounds for full \alcio SHACL, but decidability already follows from finite satisfiability of full hybrid $\mu$-calculus being decidable. Before discussing this, we first consider the finite model property under the well-founded semantics.

\begin{pr}\label{prop:fmp}
    \alco SHACL has the finite model property, \alci SHACL does not, under the well-founded semantics.
\end{pr}

\begin{proof}
    Given the result of Theorem \ref{theorem:translationiscorrect}, the finite model property for \alco SHACL follows directly from \cite{DBLP:journals/logcom/Tamura15}. To see that \alci SHACL does not have the finite model property, consider the following set of (stratified) constraints:
        $\C = \{s \gets \lnot s', s' \gets \forall r.s' \land \lnot s'', s'' \gets \forall r^-.s''\}$.
    Satisfiability of $s$ against $\C$ corresponds to satisfiability of $\final(\trpe(s)) = \nu X_{\Bar{s'}}. \langle r \rangle X_{\Bar{s'}} \land \mu Y_{s''}.[r^-] Y_{s''}$, which is a well-known example of a modal $\mu$-calculus formula that only has infinite models.
\end{proof}

\begin{theorem}\label{theorem:finitesat}
        Deciding finite shape satisfiability w.r.t.\ constraints for \alcio SHACL is decidable, under the well-founded semantics.
\end{theorem}

\begin{proof}
    Given the result of Theorem \ref{theorem:translationiscorrect}, the finite satisfiability of the hybrid $\mu$-calculus being decidable suffices. The latter follows from the finite satisfiability of guarded fixed point logic, which naturally extends the modal $\mu$-calculus with backwards modalities \cite{DBLP:journals/ipl/BaranyB12}. In guarded fixed point, given that there is no bound on the arity of predicates, we can add a prefix $x_1,\ldots,x_n$ to every predicate, where $n$ is the amount of individuals used, and use $x_i$ in place of the $i$-th individual, to model the usage of nominals \cite{DBLP:journals/jolli/CateF05}.
\end{proof}

For upper bounds, recall the translation of SHACL into $\mu$-calculus may come at the cost of a formula of size exponential in the size of the constraint set. Thus, we can combine finite satisfiability results for different fragments of the full hybrid $\mu$-calculus with the result of our translation (Theorem \ref{theorem:translationiscorrect}) by adding one exponential. Considering how nominals can be eliminated in guarded formulas as described above, at the cost of not being able to bound the arity of predicates \cite{DBLP:journals/jolli/CateF05}, it follows from the work by \citeauthor{DBLP:journals/ipl/BaranyB12} (\citeyear{DBLP:journals/ipl/BaranyB12}) that finite satisfiability for \alcio SHACL can be solved in 3-\exptime. For \alci SHACL, 2-\exptime-membership for deciding finite satisfiability already follows from work by \citeauthor{DBLP:conf/icalp/Bojanczyk02} (\citeyear{DBLP:conf/icalp/Bojanczyk02}).
We have no reason to believe these upper bounds are tight: they are obtained via multiple translations that cause one or more exponential blow-ups, which could well be avoidable in a more direct approach.

\paragraph{(Finite) Satisfiability and Implication of Documents.}

With a small update of the automaton construction, the \exptime complexity results transfer to the problem of document satisfiability.

\begin{theoremrep}\label{theorem:satisfiability-alcio}
    Deciding document satisfiability for \alcio SHACL is \exptime-complete, under the well-founded semantics. Finite document satisfiability for \alcio SHACL is decidable.
\end{theoremrep}

\begin{proof}
    The main idea is to take the 2ATA $\automaton(\C,s,\guess)$ for some $(c,s) \in \T$, and update the transition function $\delta$ in the following way. For every $(D,s) \in \T$, with $D \in \NC \cup \roles$, if $D \in \sigma$, the transition function for the state that is traversing the whole tree already, $q'$, is updated with an extra conjunct: $(0,\trpt(s))$. For the individuals $a_i \in \NI$, this can directly be taken care of by adding $(i,\trpt(s))$ as a conjunct to $\delta(q_0',\sigma)$, where $q_0'$ is the state that already ensures all nominals in the graph are assigned the right states.

    More specifically, taking the exact same automaton as in the proof of Proposition \ref{pr:2ata}, with the following being redefined as described below, suffices.

    \begin{align*}
    \Gamma(i) &= \begin{cases}
        (i,\bot) &\text{if } \gamma_i = \bot\\
        \bigwedge_{\varphi \in \gamma_i} (i,\trpt(\varphi)) \land \bigwedge_{(a_i,s)\in\T} (i,\trpt(\varphi)) &\text{if } \gamma_i \subseteq \sub(\C) \cup \{\lnot s \mid s \in \sub(\C)\}
    \end{cases}\\
    \delta(q',\sigma) &= \bigwedge_{\xrightarrow{r^-}a_i \in \sigma, \forall r.s \in \gamma_{f(i)}} (0,\trpt(s)) 
    \land \bigwedge_{(D,s)\in\T,D \in (\NC \cup \roles)\cap \sigma} (0,\trpt(s))\\
    &\qquad\land \bigwedge_{\xrightarrow{r^-}a_i \in \sigma, \{\exists r.s, \lnot s'\} \subseteq \gamma_{f(i)},s' \gets \exists r.s \in \C} (0,\trmt(s)) \land \bigwedge_{i = 1}^k ((i,q') \lor (i,\bot))
\end{align*}
For finite satisfiability, we note decidability follows from among others the combination of Proposition \ref{pr:docsat-docimp}, Theorem \ref{thm:docImplicReduc} combined with reasoning analogous to the proof of Theorem \ref{theorem:finitesat}.
\end{proof}

A direct consequence of this result is that for \alco SHACL deciding finite satisfiability \wrt constraints under the well-founded semantics is \exptime-complete, 
since the fragment 
has the finite model property; see Proposition \ref{prop:fmp}.

As the complement of an 2ATA can be constructed in polynomial time \cite{DBLP:journals/tcs/MullerS87}, it follows that deciding whether $(\C,\T)$ implies $(\C',\T')$ under the well-founded semantics for \alcio SHACL can be decided by taking the conjunction of the automaton checking document satisfiability for $(\C,\T)$ and the complement of the automaton recognising document satisfiability for $(\C',\T')$.

\begin{theorem}
    Deciding document implication for \alcio SHACL is \exptime-complete, under the well-founded semantics. Finite document implication for \alcio SHACL is decidable.
\end{theorem}

\paragraph{Least Fixed Point Semantics.} Apart from the supported, stable model and well-founded semantics for recursive SHACL, another semantics proposed in the literature is the least fixed point semantics for stratified SHACL \cite{ecai2023,aij2026}. As the name already suggests, this is a semantics only using a least fixed point in its definition, making it conceptually easier. Nonetheless, this suffices to compute either the well-founded or stable model semantics over stratified constraints; if $\C$ is stratified, these three semantics coincide. In particular, this means the discussed static analysis problems for least fixed points semantics can be simply copied from the well-founded setting discussed above. Moreover, the translation of stratified SHACL, and thus also of SHACL under the least fixed point semantics, would fall in the alternation free $\mu$-calculus. To see this, in a nutshell, note that in a $\mu$-calculus formula with alternation like $\mu X_s.\varphi$ such that $\nu X_{\Bar{s’}}.\varphi(X_s) \in \sub(\varphi)$ we find $s’$ depending negatively on $s$; because of the $\nu$-formula appearing inside the $\mu$-formula, and $s$ depending negatively on $s’$; because of $X_s$ appearing inside the $\nu$-subformula. Clearly, this makes the constraint set unstratified. We are optimistic that this can be exploited to achieve more scalable static analysis results.
\section{Conclusion and Discussion}

In this paper, by establishing a connection to the full hybrid $\mu$-calculus, we provide the first decidability and complexity results for basic static analysis problems in recursive SHACL under the well-founded semantics. These results contrast sharply with the undecidability results we show for the supported and stable model semantics. A key observation is that the implication problem becomes computationally unmanageable when the underlying semantics can assign multiple shape assignments to a graph. The well-founded semantics yields a unique shape assignment for every graph, which forms one of the bases under our decidability result.

There are several directions for future work. First, in the context of practical applications, it is relevant to understand  static analysis in the presence of a richer syntax for shape expressions. For instance, adding numeric restrictions (such as qualified number restrictions in DLs) to SHACL under the well-founded semantics is a natural next step. This requires understanding how a well-founded semantics can be defined in this setting and determining how our algorithmic methods can be adapted to it. Second, our translation from SHACL under the well-founded semantics into the full hybrid $\mu$-calculus raises the question of whether a similar translation exists for Datalog with negation under the well-founded semantics into Fixpoint Logic. A natural follow-up question is whether decidability results for Guarded Fixpoint Logic can yield new positive results for restricted Datalog fragments under the well-founded semantics.

\section*{Acknowledgments}
\imge \quad  The project leading to this application has received funding from the European Union's Horizon 2020 research and innovation programme under grant agreement No 101034440.  

In addition, this work was partially funded by the Austrian Science Fund (FWF) projects P30873, PIN8884924, and 10.55776/COE12.

\bibliographystyle{kr}
\bibliography{bibliography}

\end{document}